%% file: journal-PART1.tex
\documentclass[journal]{IEEEtran}

\input{preamble}

\begin{document}
\colorlet{myred}{black}

\input{tex_parts/title_authors_thanks}
\maketitle

\input{tex_parts/abstract}

\input{tex_parts/introduction}

\input{tex_parts/body}
\input{tex_parts/appendix}
\bibliographystyle{IEEEtran}
\bibliography{biblio_poaLP}

\input{tex_parts/biography}
\vfill\break

\end{document}

%% file: preamble.tex
\usepackage{graphicx}
\usepackage{color}          %
\usepackage{amsfonts}       %
\usepackage{amsthm}         %
\usepackage{mathtools}      %
\graphicspath{{figures/}}

\definecolor{Granata}{rgb}{0.64,0,0} 

\newcommand{\revisioned}{}
\newcommand{\mc}[1]{\mathcal{#1}}
\newcommand{\mb}[1]{\mathbb{#1}}
\newcommand{\N}{N}

\newcommand{\poa}{{\rm PoA}}
\newcommand{\nashe}[1]{{\rm NE}({#1})}
\newcommand{\be}{\begin{equation}}
\newcommand{\ee}{\end{equation}}

\renewcommand{\ae}{a^{\rm ne}}

\renewcommand{\a}[1]{a^{#1}}
\newcommand{\aopt}{a^{\rm opt}}
\newcommand{\ami}{a_{-i}}
\newtheorem{definition}{Definition}

\newcommand{\fee}{{\cal F}}
\newcommand{\wee}{{\cal W}}
\newcommand{\aee}{{\cal A}}
\newcommand{\gee}{{\cal G}}
\newcommand{\ree}{{\cal R}} 
\newcommand{\tee}{{\cal T}} 
\newcommand{\arr}{\mathbb{R}}

\DeclareSymbolFont{bbold}{U}{bbold}{m}{n}
\DeclareSymbolFontAlphabet{\mathbbold}{bbold}

\usepackage{enumitem}
\usepackage{cite}
\usepackage{tikz}
\usepackage{pgfplots}
\usetikzlibrary{intersections,backgrounds,plotmarks}

\newtheorem{theorem}{Theorem}

\newtheorem{remark}{Remark}
\newtheorem{example}{Example}
\newtheorem*{example_rev}{Example \ref{ex:vta} revisited}
\newtheorem{corollary}{Corollary}

\newtheorem{lemma}{Lemma}
\def\vspacesteps{1.5mm}

\def\myspaceintro{1.5mm}
\def\myspaceproofs{1.5mm}
\def\eqspacezero{\hspace*{-0.3mm}}
\def\eqspace{\hspace*{-0.26mm}}
\def\eqspacetwo{\hspace*{-0.08mm}}
\def\eqspacethree{\hspace*{-0.26mm}}
\def\eqspacefour{\hspace*{-0.15mm}}
\def\eqspacefive{\hspace*{-0.3mm}}
\newcommand{\cdotshort}{\!\cdot\!}

\newcommand{\fopt}{f_{\rm opt}}
\newcommand{\muopt}{\mu_{\rm opt}}
\newtheoremstyle{break}
  {\topsep}{\topsep}%
  {\itshape}{}%
  {\bfseries}{}%
  {\newline}{}%
\theoremstyle{break}

\usepackage{comment}
\excludecomment{mycomment}
\newcommand{\cut}[1]{}
\begin{mycomment}
\renewcommand{\cut}[1]{#1}
\end{mycomment}
\usepackage{soul}

\newcommand{\poas}{{\rm RPoA}}
\newcommand{\Ir}{\mathcal{I}_R}
\newcommand{\geefeewee}{\gee_{\fee,\wee}}
\newcommand{\geefw}{\gee_{f,w}}
\newcommand{\poafeewee}{\poa(\fee,\wee)}
\newcommand{\feewee}{\fee,\wee}
\newcommand{\hatgeefw}{\hat{\gee}_{f,w}}
\renewcommand{\phi}{\Phi}
\renewcommand{\k}{k}
\usepackage{nccmath}

\DeclareMathOperator*{\argmax}{arg\,max}
\DeclareMathOperator*{\argmin}{arg\,min}
\usepackage{balance}

%% file: tex_parts/title_authors_thanks.tex
\title{\LARGE \bf Utility Design for Distributed Resource Allocation -- Part I: \\
Characterizing and Optimizing the Exact Price of Anarchy
}
\author{Dario~Paccagnan,~Rahul~Chandan~and~Jason~R.~Marden
\thanks{This research was supported by SNSF Grant \#P2EZP2-181618 and by ONR Grant \#N00014-15-1-2762, NSF Grant \#ECCS-1351866. D. Paccagnan is with the 
Mechanical Engineering Department and the Center of Control, Dynamical Systems and Computation, UC Santa Barbara, CA 93106-5070, USA. R. Chandan and J.\,R. Marden are with the Department of Electrical and Computer Engineering, University of California, Santa Barbara, CA 93106-5070, USA. Email: {\tt\{dariop,rchandan,jrmarden\}@ucsb.edu} .}
}

%% file: tex_parts/abstract.tex
\begin{abstract}
Game theory has emerged as a fruitful paradigm for the design of networked multiagent systems. A fundamental component of this approach is the design of agents' utility functions so that their self-interested maximization results in a desirable collective behavior. In this work we focus on a well-studied class of distributed resource allocation problems where each agent is requested to select a subset of resources with the goal of optimizing a given system-level objective.
Our core contribution is the development of a novel framework to tightly characterize the worst case performance of any resulting Nash equilibrium (price of anarchy) as a function of the chosen agents' utility functions. 
 Leveraging this result, we identify how to design such utilities so as to optimize the price of anarchy through a tractable linear program.  
 This provides us with a priori performance certificates applicable to any existing learning algorithm capable of driving the system to an equilibrium.
 Part II of this work specializes these results to submodular and supermodular objectives, discusses the complexity of computing Nash equilibria, and provides multiple illustrations of the theoretical~findings.
\\[3mm]
\end{abstract}

%% file: tex_parts/introduction.tex
\section{Introduction}
\IEEEPARstart{M}{ultiagent} systems have enormous potential for solving many of the current societal challenges. Robotic networks can operate in post-disaster environments and reduce the impact of industrial or natural calamities \cite{kuntze2012seneka, kitano1999robocup}. They have the potential to increase the food production and to reduce the water consumption \cite{kozai2015plant}. 
Fleets of autonomous cars are forecasted to revolutionize the mobility, to reduce traffic congestion as well as pollutant emissions \cite{spieser2014toward}.
On a smaller scale, swarms of medical ``microbots'' promise groundbreaking results by means of local drug delivery \cite{servant2015controlled} or microsurgery \cite{ishiyama2002magnetic}.

The \emph{main challenge} in controlling such systems is to design \emph{local decision rules} for the individual subsystems to guarantee that the collective behavior is desirable with respect to a \emph{global objective} \cite{li2013designing}. The spatial distribution, privacy requirements, scale, and quantity of information associated with typical networked systems do not allow for centralized communication and decision-making, but require instead the use of \emph{distributed protocols}.  That is, the agents in the system must make decisions independently in response to available information.  
Such problems are typically posed as optimization problems (finite or infinite dimensional), where the system-level objective is captured by an \emph{objective function}, while physical laws and informational availability are incorporated as \emph{constraints} on the decision variables \cite{Raffo, CortesBullo}. 
One common approach for deriving admissible algorithms is to distribute existing centralized optimization schemes by leveraging the structure of the given problem, e.g., distributed gradient ascent, primal-dual and Newton's method, among others \cite{nedic2009distributed, wei2013distributed}. 

An alternative approach, termed \emph{game design}, has established itself as a valuable set of tools to complement these more traditional techniques \cite{shamma2007cooperative}. Rather than directly specifying a decision-making process, local utility functions are assigned to the agents, so that their self-interested optimization translates to the achievement of the system-level objective; see Figure \ref{fig:gamedesign} for an overview of the approach. The motivation for studying equilibria of such games, as opposed to equilibria of a dynamical process, stems from the existence of readily available distributed learning algorithms that can be utilized to drive the collective behavior to a Nash equilibrium in a given game \cite{shamma2007cooperative,marden2013distributed,blume1993statistical, fudenberg1998theory, marden2012revisiting, gentile2017nash, paccagnan2017coupl, yi2017distributed}.
\revisioned{When the agent utility functions are set to be  equal to the system-level objective, the resulting algorithms and performance guarantees resemble the centralized optimization methodologies discussed above. The value of the game design approach emerges when the designed agent utility functions are \emph{not} equal to the system-level objective.  In fact, we demonstrate the non-intuitive fact that having agents respond to local utility functions that do not reproduce the system-level objective significantly improves the efficiency guarantees of the resulting collective behavior. }

\begin{figure}[h!]
\centering
\vspace*{-2mm}
\includegraphics[scale=0.67]{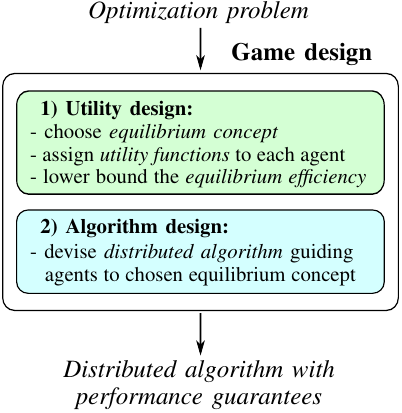}
\vspace*{-2mm}
\caption{Game theoretic approach for the design of distributed algorithms. }
\vspace*{-2mm}
\label{fig:gamedesign} 
\end{figure}

The core of this work centers on characterizing the relationship between the agents' utility functions and the efficiency of the emergent collective behavior, that is, this work focus on the \emph{utility design} step depicted in Figure \ref{fig:gamedesign}.  We model the emergent collective behavior as a pure Nash equilibrium of the game defined by the agents' utility functions, and 
 measure the corresponding efficiency with the notion of price of anarchy \cite{Koutsoupias}. The price of anarchy provides performance guarantees associated with the worst-performing Nash equilibrium of the designed game, relative to the optimal performance.

\noindent\revisioned{{\bf Related Works.}}
There has been extensive research in the field of algorithmic game theory focused on analyzing the price of anarchy, as well as other efficiency measures such as the price of stability \cite{Koutsoupias,schulz2003performance} and price of stochastic anarchy \cite{chung2008price}. However, most of the results are purely analytical and do not properly address the design questions considered in this manuscript. One noteworthy result involves the widely studied smoothness framework, which provides a general approach  - in the form of an inequality constraint involving the utility functions - to bound the price of anarchy \cite{roughgarden2009intrinsic}. Inspired by this line of work, \cite{nadav2010limits,bilo2012unifying,kulkarni2014robust,thang2017game} developed mathematical programs aimed at deriving bounds on the price of anarchy by incorporating this constraint. Unfortunately, we show in Section \ref{sec:smooth} that the price of anarchy bounds associated with the smoothness framework are relevant only when the agents' utility functions are budget-balanced, i.e., the sum of the agents' objective is equal to the system-level objective.  While this constraint is well justified for a number of problems modeled through game theory (e.g., cost sharing games \cite{moulin2001strategyproof}), it has little bearing on the design of local utility functions in multiagent systems, as studied here. 

Much less is known about tight price of anarchy guarantees outside of the case of budget-balanced objectives, with the exception of some specific problem domains including coverage problems \cite{gairing2009covering} and resource allocation problems with convex costs or concave welfare \cite{phillips2017design, jensen2018}. 
We note that all the aforementioned approaches based on the reformulation of the smoothness condition through mathematical programming techniques \cite{nadav2010limits,bilo2012unifying,kulkarni2014robust,thang2017game} result in optimization problems whose size is exponential in the number of agents and in the number of strategies available to each agent. Some of these works even require prior knowledge on the structure of the worst case instances, e.g., \cite{bilo2012unifying}.  
On the contrary our approach produces a linear program whose size does not depend on the number of agents' strategies, and grows only linearly in the number of agents. Just as importantly, our linear program automatically discovers the structure of the worst case instances.

\vspace*{\myspaceintro} 
\noindent {\bf Contributions.} 
This paper generalizes the afore-mentioned application specific results by developing methodologies for game design in a well-studied class of distributed resource allocation problems, where each agent selects a subset of resources with the goal of maximizing a system-level objective function that is separable over the resources. %
The main contributions of this paper include the following:

\setlist[enumerate]{leftmargin=*}
\begin{enumerate}
	\item We show that the smoothness framework typically used to bound the price of anarchy is not suited for the utility design problems considered, as the  corresponding efficiency bounds are conservative (Theorem \ref{prop:limitations}).
	\item We resolve the problem of computing the exact (i.e., tight) price of anarchy by means of a tractable linear program in its primal and dual form (Theorems \ref{thm:primalpoa} and \ref{thm:dualpoa}). The latter program features only $2$ scalar decision variables and  $\mc{O}(n^2)$ constraints, where $n$ represents the number of agents. Such a program can be solved efficiently.%
	\item We solve the problem of designing agent utility functions so as to optimize the resulting price of anarchy. We show that this problem can be posed as a tractable linear program in $n+1$ variables and $\mc{O}(n^2)$ constraints (Theorem~\ref{thm:optimizepoa}). 
\end{enumerate}

Part II demonstrates the breadth of the approach by specializing these results to the class of submodular and supermodular resource allocation problems. In this context, we show how our approach subsumes and generalizes existing fragmented results. We conclude Part II by showcasing the applicability of our techniques by means of two applications.

\vspace*{\myspaceintro}
\noindent {\bf Organization.} Section \ref{sec:modelandmetric} contains the model, the game theoretic approach, and the corresponding performance metrics. 
Section \ref{sec:smooth} shows the inapplicability of the smoothness framework to the utility design problem considered.
Sections \ref{sec:charactandopt} and \ref{sec:optimalpoa} show how to reformulate the problems of characterizing and optimizing the price of anarchy as tractable linear programs.

\vspace*{\myspaceintro} 
\noindent {\bf Notation.} 
We use $\mb{N}$, $\mb{R}_{>0}$ and $\mb{R}_{\ge0}$ to denote the set of natural numbers, positive and non-negative real numbers.
For any $p,q\in\mb{N}$, $p\le q$, let $[p]=\{1,\dots,p\}$ and $[p, q]=\{p,\dots,q\}$.
 Given a finite set $\mc{I}$, $|\mc{I}|$ denotes its cardinality; $e$ represents Euler's number. All proofs are reported in the Appendix. 

%% file: tex_parts/body.tex
\section{Model and performance metrics}
\label{sec:modelandmetric}
\subsection{Problem formulation}
\label{sec:problemformulation}
In this paper we consider a framework for distributed resource allocation. Let $N = \{1, 2, \dots, n \}$ be a set of agents, and $\mc{R}=\{r_1,\dots,r_m\}$ be a set of resources, with $n, m\in\mb{N}$. Each resource $r \in \ree$ is associated with a local welfare function $W_r : [n] \rightarrow \mathbb{R}$ that captures the welfare accrued at each resource as a function of the utilization, i.e., $W_r(j)$ is the welfare generated at resource $r$ if there are $j$ agents utilizing that resource.  Finally, each agent $i \in N$ is associated with an admissible choice set $\aee_i \subseteq 2^\ree$.  The goal of the system designer is to find an allocation $a = (a_1, a_2, \dots, a_n) \in \aee = \aee_1 \times \dots \times \aee_n$ optimizing a system-level objective of the form
\begin{equation}\label{eq:welfaredef}
W(a)=\sum_{r \in \cup_i a_i} W_r(|a|_r),
\end{equation}
where $|a|_r$ denotes the number of agents choosing resource $r$ in allocation $a$, i.e., the cardinality of the set $\{i \in N: r \in a_i\}$.  We will often use $\ami=(a_1,\dots,a_{i-1},a_{i+1}, \dots, a_n)$ to denote the decision of all the agents but $i$.

\revisioned{
\begin{example}[Vehicle-target assignment]\label{ex:vta}
	Consider the classic vehicle-target assignment problem \cite{murphey2000target}. In this problem there are a set ${\cal T}$ of targets, each associated to its importance $v_t \geq 0$, and a set of agents $N = \{1, 2, \dots, n\}$. Each agent $i \in N$ is given a set of possible assignments $\aee_i \subseteq 2^{\tee}$, and - in the homogeneous version - a common success probability $p \in [0,1]$. The goal is to determine an admissible vehicle-target assignment $a \in \aee$ to maximize the value of acquired targets, as expressed by	%
	\begin{equation}\label{eq:wta}
	W(a) = \sum_{t \in \cup_i a_i} v_t \left( 1- (1-p)^{|a|_t} \right). 
	\end{equation}
\end{example}
	\begin{example}[Weighted maximum coverage]\label{ex:wmc}
	An important case of the vehicle target assignment problem that warrants independent attention is the weighted maximum coverage problem \cite{nemhauser1978analysis}, obtained setting $p=1$ in \eqref{eq:wta} and assuming all agents have identical assignments' sets $\mc{A}_i=\bar{\mc{A}}$. 
\end{example}
\begin{example}[Routing through a shared network]\label{ex:routing}
	Consider the classic routing problem given in \cite{roughgarden2007routing}. In this problem a set $N$ of users utilizes a common network comprised of a set of edges $E$, where each edge $e \in E$ is associated to a cost function $c_e: \{1, \dots, n \} \rightarrow \arr$.  The term $c_e(j)$ captures the quality of service (or latency) on edge $e$ if there are $j\geq1$ users sharing that edge.   Each user $i \in N$ is associated with a given source and destination, which defines a set $\aee_i \subseteq 2^E$ of admissible paths. %
	The goal of the system designer is to determine an allocation $a \in \aee$ of users to the network to~optimize~the~total~congestion 
	\begin{equation}\label{eq:cong}
	W(a) = -\sum_{e \in\cup_i a_i } |a|_e \cdot c_e(|a|_e). 
	\end{equation}
\end{example}
}

\noindent
From this point on, we consider resource allocation problems where the welfare functions satisfy $W_r(j) \geq 0$ for all $j \geq 1$, and thus $W(a)\geq 0$ for all $a \in \aee$. The case of $W(a)\leq 0$ for all $a \in \aee$ falls under the framework of cost minimization. While we do not explicitly delve into the framework of cost minimization here, all of the forthcoming results have an analogous result that holds for the cost minimization setting.  
\subsection{Local utility design and the price of anarchy} \label{sec:metrics}
{Since finding a feasible allocation maximizing \eqref{eq:welfaredef} is an intractable problem \cite{nemhauser1978analysis},
 we focus on deriving efficient and \emph{distributed} algorithms for attaining approximate solutions to the maximization of \eqref{eq:welfaredef}. %
In this respect, each agent $i\in N$ is tasked to make independent choices in response to pieces of information only regarding resources he can select, i.e.,~$r_i\in\mc{A}_i$.}

Rather than directly specifying a decision-making process, we adopt the framework of game theory and utility design.  Here, each agent $i \in N$ is associated with a utility function of the form $U_i : \aee \rightarrow \arr$ that guides their individual behavior. We focus on the class of \emph{local} agent objective functions where for any agent $i \in N$ and allocation $a\in \aee$ we let
\begin{equation}
U_i(a) = \sum_{r\in a_i} f_r(|a|_r), \label{eq:utilities}
\end{equation}
and $f_r : \{1, \dots, n\} \rightarrow \arr$ defines the utility each agent receives at resource $r$ as a function of the number of agents selecting that resource in the allocation $a$.  We refer to $\{f_r\}_{r \in \ree}$ as the utility generating functions since each agent's utility is fully determined once $\{f_r\}_{r \in \ree}$ are specified.
  We denote one such game with the \revisioned{tuple $G = (N, \ree, \aee, \{W_r\}_{r \in \ree}, \{f_r\}_{r\in\ree})$.}\footnote{\revisioned{An apparently less restrictive assumption is that of letting $f_r(|a|_r)$ depend on which agent $i$ is currently selecting resource $r$. It is nevertheless possible to show that this additional degree of freedom will not yield any improvement in the equilibrium efficiency, see e.g.~\cite{gairing2009covering}.}}

The core of this work is on analyzing the efficiency of the equilibria associated with the game generated by $\{f_r\}_{r \in \ree}$ according to the utility functions as defined in \eqref{eq:utilities} and the system-level objective introduced in \eqref{eq:welfaredef}. 
We focus on the notion of Nash equilibrium, which is guaranteed to exist for any game $G$ with utility functions \eqref{eq:utilities}, thanks to the fact that $G$ is a congestion game~\cite{rosenthal1973class}. 
\begin{definition}[Nash equilibrium, \cite{Nash01011950}]
	An allocation $\ae \in \mc{A}$ is a pure Nash equilibrium if $U_i(\ae)\ge U_i(a_i,\ae_{-i})$ for all
	alternative allocations $a_i\in \mc{A}_i$ and for all agents $i\in N$.
\end{definition}

{We characterize the efficiency of a Nash equilibrium - which we refer to as simply an equilibrium - using the notion of price of anarchy ($\poa$). The price of anarchy of an instance $G$ is defined as the ratio between the welfare at the worst-performing Nash equilibrium, and the maximum attainable welfare \cite{Koutsoupias}.\footnote{%
\revisioned{
The choice of Nash equilibrium provides us with potentially better performance guarantees compared to what offered by more permissive - but easier to compute - equilibrium notions, e.g., coarse correlated equilibrium \cite{aumann1987correlated}. Indeed, since every Nash equilibrium is also a coarse correlated equilibrium, the worst-performing coarse correlated equilibrium yields a system-level objective that is no better than that of the worst-performing Nash equilibrium. The drawback of this choice stems from the intractability of Nash equilibria, which are hard to compute ($\mathcal{PLS}$-complete, \cite{fabrikant2004complexity}) even for the class of congestion games to which $G$ belongs to. Nevertheless, under structural assumptions on the sets $\{\mc{A}_i\}_{i\in\N}$ similar to those used in combinatorial optimization, computing a Nash equilibrium is a polynomial task (See Proposition 2 in Part II).
Finally, the guarantees offered by Nash equilibria are deterministic.
}
}
Price of anarchy guarantees are particularly impactful when extended from a single game $G$ to a family of games $\gee$; however, defining a family of games requires defining the utility generating functions for each game instance. Here, we focus on the case where a system designer is unaware of the exact number of agents, the number of resources, and the agents' action sets, while the only information available a priori is the set $\wee$ describing the possible welfare function utilized, i.e., $W_r\in\wee$ for all resources.  This request stems from the observation that the previous pieces of information may be unreliable, or unavailable to the system designer due to, e.g., communication restrictions or privacy concerns.

Given this uncertainty, the system designer commits a priori to a specific utility generating function for each welfare functions $W_r \in \wee$, which we express by $f_r = \fee(W_r)$. The realized resource allocation problem merely employs the utility generating functions $\{\fee(W_r)\}_{r\in\ree}$.  The map $\fee$ constitutes our design choice, and we refer to it as the utility generating mechanism. We denote with $\geefeewee$ the set of games induced by $\fee$, i.e., \mbox{any game $G \in \geefeewee$ of the form}}
\begin{equation}
G = (N, \ree, \mc{A}, \{W_r\}_{r\in\ree}, \left\{\fee(W_r)\right\}_{r \in \ree}),
\end{equation}
where $N$ is any set of agents, $\ree$ is any set of resources, $\mc{A}$ is any allocation set, and $\{W_r\}_{r\in\ree}$ is any tuple of functions satisfying $W_r\in\mc{W}$. 
The price of anarchy of the family of games $\geefeewee$ is defined as the worst case over $G\in\geefeewee$,~i.e.,
\begin{equation}
\label{eq:poadef}
\poafeewee = \inf_{G \in \geefeewee} \left(
 \frac
 { \min_{a \in  \nashe{G}} W(a)}
 {\max_{a \in  \aee} W(a)}
\right).
\end{equation}
where $\nashe{G}$ denotes the set of equilibria of $G$. While the function $W$ also depends on the instance $G$ considered, we do not indicate it explicitly, to simplify notation. The quantity $\poafeewee$ characterizes the efficiency of the worst-performing equilibrium relative to the corresponding optimal allocation, over all possible instances in $\geefeewee$. In the non-degenerate cases where $\max_{a\in\mc{A}}W(a)>0$, it holds $0\le\poafeewee\le 1$, and the higher the price of anarchy, the better performance certificates we can offer. %
Observe that when an algorithm is available to compute one such equilibrium, the price of anarchy also represents the \emph{approximation ratio} of the corresponding algorithm over all the instances in $\geefeewee$. 
\revisioned{
\begin{example_rev}[Vehicle-target assignment]
	Consider the vehicle-target assignment problem of Example~\ref{ex:vta}. Define $\wee$ as the set containing only welfare functions of the form
	\begin{equation}
	W_r(j) = v_r \cdot(1-(1-p)^j),\quad \forall j \in \mathbb{N}.
	\end{equation}
	for all $v_r \geq 0$, $p\in [0,1]$.
	Further, consider a utility generating mechanism $\fee^{\rm mc}$, known as the marginal contribution, where for any $W_r \in \wee$ we have $f^{\rm mc}_r = \fee^{\rm mc}(W_r)$ with 
	\begin{equation}\label{eq:ugm}
	f^{\rm mc}_r(j) = v_r \cdot p \cdot (1-p)^{j-1}, \quad \forall j \in \mathbb{N}.
	\end{equation}
	Note that the utility generating mechanism provided in \eqref{eq:ugm} results in a well-defined game for any instance of the vehicle-target assignment problem given Example~\ref{ex:vta}.  %
	  Furthermore, it is shown in \cite{marden2013distributed} that ${\rm PoA}(\fee^{\rm mc}, \wee) = 1/2$, meaning that regardless of the underlying vehicle-target assignment problem, all equilibria are guaranteed a performance within $50\%$ of optimal. This paper will develop a framework to design the best possible utility generating mechanism, and accompany that with a performance certificate significantly beyond $50\%$.
\end{example_rev}
}
\noindent
We decompose the utility design problem in two tasks: 
\begin{itemize}
\item[i)] providing a bound (or ideally an exact characterization) of the price of anarchy as a function of $\fee$ and $\wee$;
\item[ii)] optimizing this expression over all mechanisms $\fee$.
\end{itemize}
 In Section \ref{sec:charactandopt}  we address i), while in Section \ref{sec:optimalpoa} we turn the attention to ii). Before doing so, we show that the existing smoothness approach is unsuitable for this purpose.

\section{Smoothness and its limitations}
\label{sec:smooth}

There has been significant research attention geared at analyzing the price of anarchy for various classes of games. One approach that is commonly employed for this purpose is termed \emph{smoothness} \cite{roughgarden2009intrinsic}. The framework of smoothness provides a technique to bound the price of anarchy of a given game $G$ by devising parameters $\lambda,\mu\ge0$ that satisfy 
\begin{equation}
\label{eq:smoothcondition}	
\sum_{i\in\N} U_i(a'_i,a_{-i})\ge \lambda W(a')-\mu W(a), \quad \forall a,a'\in\mc{A}.
\end{equation}
We refer to a game fulfilling \eqref{eq:smoothcondition} as $(\lambda,\mu)$-smooth.  If a game $G$ is $(\lambda,\mu)$-smooth and $\sum_i U_i(a) \leq W(a)$ for all $a\in \aee$, \cite{roughgarden2009intrinsic} proves that the price of anarchy of $G$ is lower bounded~by
	\begin{equation}
 \frac
 { \min_{a \in  \nashe{G}} W(a)}
 {\max_{a \in  \aee} W(a)}\ge \frac{\lambda}{1+\mu}\,. 	
\label{eq:lambdamubound}
\end{equation}
A similar argument can be used to lower bound the price of anarchy for the family $\geefeewee$ introduced in Section \ref{sec:modelandmetric}.  In this respect, for any given utility generating mechanism $\fee$, the best bound on the price of anarchy \eqref{eq:poadef} that can be derived via smoothness is given by the solution to the following  program 
\begin{equation}\label{eq:robust-poa}
\small
\poas(\feewee) \!=\! \sup_{\lambda,\mu \ge0} \!\left\{\frac{\lambda}{1+\mu}~ \text{s.t.} ~(\lambda,\mu) ~\text{satisfy \eqref{eq:smoothcondition}},  \forall G \in \geefeewee \right\}.
\end{equation}
The term $\poas(\feewee)$ lower bounds $\poa(\feewee)$, and is often referred to as the robust price of anarchy \cite{roughgarden2009intrinsic}.
Note that the smoothness framework forces us to restrict the attention to mechanisms whose corresponding utilities satisfy 
\be\sum_{i\in\N} U_i(a)\le W(a),\quad \forall a\in\mc{A}, \forall G\in\geefeewee,
\label{eq:subbudget}
\ee
 else no guarantee is provided by \cite{roughgarden2009intrinsic}.
 Thus, in the remaining of this section \emph{only}, we consider utilities satisfying \eqref{eq:subbudget}. Finally, we refer to mechanisms whose corresponding utilities satisfy \eqref{eq:subbudget} with equality as \emph{budget-balanced} mechanisms.

At first glance it appears that the smoothness framework could be extremely beneficial for characterizing the price of anarchy associated with different utility generating mechanisms.  Unfortunately, the following proposition demonstrates a significant weakness associated with this framework.  

\begin{theorem}[{\bf Limitations of the smoothness framework}]
\label{prop:limitations}
\begin{itemize}
\item[]
\item[i)] 
The budget-balanced mechanism $\fee^{\rm es}(W_r)=f^{\rm es}$, with $f^{\rm es}(j) = W_r(j)/j$, $j\in \mb{N}$, provides the best bound on the price of anarchy attainable using a smoothness argument.
\\[-3mm]
\item[ii)] 
There exist classes of problems for which $\fee^{\rm es}$ does not optimize the price of anarchy. 
Indeed, consider the weighted maximum coverage problem of Example \ref{ex:wmc}, i.e., let $\wee$ contain only functions   of the form $W_r(j)=v_r$ for any $v_r\ge 0$, $j\in\mb{N}$, and compare the mechanisms $\fee^{\rm es}(W_r)=f_r^{\rm es}$ with $\fee^\star(W_r)=f_r^\star$, where
\[
\begin{aligned}
f^{\rm es}_r(j) &= v_r \cdot\frac{1}{j},\quad &&\forall j\in \mb{N},\\
f^\star_r(j) &= v_r \cdot\frac{(j-1)!}{e-1}\cdot\left(e-\sum_{i=0}^{j-1}\frac{1}{i!}\right),\quad &&\forall j\in \mb{N}.
\end{aligned}
\]
For any mechanism $\fee$, the best bound on the price of anarchy that can be derived with a smoothness argument~is
\[
\poas(\fee,\wee)\le \poas(\fee^{\rm es},\wee)=\frac{1}{2}.
\]
Nevertheless, $\fee^\star$ has better price of anarchy than $\fee^{\rm es}$,~as 
\[
\poa(\fee^\star,\wee)=1-\frac{1}{e} > \frac{1}{2}=\poa(\fee^{\rm es},\wee).
\]
\end{itemize}
\end{theorem}

Limited to mechanisms where the smoothness framework can be applied, the first claim shows that $\mc{F}^{\rm es}$ gives the best bound on the price of anarchy (i.e., optimizes $\poas(\fee,\wee)$). 
The second claim demonstrates that such conclusion \emph{does not} carry over to the \emph{true} price of anarchy \eqref{eq:poadef}, thereby highlighting a significant discrepancy between the robust price of anarchy and the price of anarchy outside the budget-balanced regime. Indeed, the utilities corresponding to $\fee^\star$ satisfy \eqref{eq:subbudget} with strict inequality, but the mechanism $\fee^\star$ has a significantly better price of anarchy than that of $\fee^{\rm es}$, unlike predicted by the smoothness framework.
\section{Characterizing the price of anarchy}
\label{sec:charactandopt}

In this section we develop a novel framework to characterize the price of anarchy in both budget-balanced and non-budget-balanced regimes. Specifically, we show how to \emph{compute} the price of anarchy through a tractable linear program~(LP).%

In the forthcoming presentation we focus on distributed resource allocation problems where the system-level objective is as in \eqref{eq:welfaredef} and the local welfare functions are of the form 
\be
\label{eq:welfarebasis}
W_r(j)=v_r\cdot w(j),\qquad w(j)>0,\quad\forall j\in N,
\ee
with $v_r\ge 0$. The function $w:\{1,\dots,n\}\rightarrow\mb{R}_{> 0}$ is \emph{fixed}, and referred to as welfare basis function. Since $w$ associates a positive real number to every integer number $\{1,\dots,n\}$, we often denote $w$ as a vector in $\mb{R}^n_{>0}$. The quantity $v_r$ can be interpreted as the value of the corresponding resource $r\in\mc{R}$, while $w$ scales such value depending on how many agents selected it. In this context, the set $\wee$ contains only functions of the form \eqref{eq:welfarebasis} for all possible values $v_r\in\mb{R}_{\ge0}$.
The welfare maximization examples previously discussed are of this form, in addition to several more that are discussed in Part II.

Given local welfare functions as in \eqref{eq:welfarebasis}, we focus on mechanisms that are linear in their argument, i.e., for which
\be
\label{eq:linmech}
\fee(v_r\cdot w)=v_r\cdot \fee(w),\quad\forall v_r\ge0.
\ee
The linearity of $\fee$ is well motivated by the observation that non linear mechanisms can only lower the price of anarchy, see \cite{chandan2019taxes}.
We denote $
f=\fee(w)$, where $f:\{1,\dots,n\}\rightarrow\mb{R}$ is hereafter our only design choice. Since $f$ associates a real number to every integer in $\{1,\dots,n\}$, we often denote $f$ as a vector in $\mb{R}^n$.%

\begin{definition}\label{ass}
\revisioned{
		\mbox{\!We\! define $\geefw^n\!$ as the set of games \!$G\!\in\!\geefeewee$ s.t.}
		\begin{itemize}
		\item[i)] $\wee$ contains only functions of form \eqref{eq:welfarebasis}, $\fee$ is as in \eqref{eq:linmech}; 
		\item[ii)] the number of agents is upper bounded by $|N|\le n$;
		\item[iii)] the optimum value satisfies $\max_{a\in\mc{A}}W(a) > 0$.
		\end{itemize}
We denote the price of anarchy of the class $\geefw^n$ with 
		\[
		\poa(f,w,n),
		\]
		as it is completely determined by $f$, $w$ and $n$.
}
\end{definition}

\subsection{The linear program reformulation}
\label{subsec:primallp}
We are now ready to state our first main contribution which characterizes the price of anarchy $\poa(f,w,n)$ of the set of resource allocations games $\geefw^n$ introduced in Definition \ref{ass}. \cut{While both $w$ and $f$ are defined over the domain $\{1,\dots,n\}$, we artificially set the non-valid extremum points as $f(0)=w(0)=f(n+1)=w(n+1)=0$, for notational convenience, else, e.g., $f(a+x+1)$ appearing in \eqref{eq:primalvalue} will not be defined for $a+x=n$.}  
Towards this goal, we define the set%
\[
\mc{I}\coloneqq \{(a,x,b)\in\mb{N}_{\ge0}^3~\text{s.t.}~1\le a+x+b\le n\}\,,%
\]
and write $\sum_{a,x,b}$ instead of $\sum_{(a,x,b)\in\mc{I}}$. Finally, we associate to each tuple $(a,x,b)\in\mc{I}$, the decision variable $\theta(a,x,b)\in\mb{R}$.

\begin{theorem}[\bf $ \rm{\bf PoA}$ as a linear program]
	\label {thm:primalpoa}
	Let $w\in\mb{R}^n_{>0}$ be a welfare basis function, and let $f\in\mb{R}^n$.
\begin{itemize}
\item[i)] \revisioned{If $f(1)\le0$, then $\poa(f,w,n)=0$ for any $n\in\mb{N}$.} 
	
\item[ii)] If instead $f(1)>0$, $n\in \mb{N}$, the price of anarchy is
	\begin{equation}
	\poa(f,w,n) = \frac{1}{W^\star},
	\end{equation} 
	where $W^\star$ is the (finite) value of the following (primal) linear program in the unknowns $\{\theta(a,x,b)\}_{(a,x,b)\in\mc{I}}$,
	\be
		\label{eq:primalvalue}
	\begin{split}
		 W^\star \!\!=\!\! \max_{\theta(a,x,b)}\!&\sum_{a,x,b} w(b+x)\theta(a,x,b)\\
		\text{s.t.}& \sum_{a,x,b}\eqspacefive[af(a \eqspacefive+\eqspacefive x)\eqspacefive-\eqspacefive bf(a \eqspacefive+ \eqspacefive x +\eqspacefive 1)]\theta(a,x,b)\! \ge\! 0, \\
		&  \sum_{a,x,b}w(a+x)\theta(a,x,b)=1,\\
		& \,\theta(a,x,b)\ge 0,\quad\forall (a,x,b)\in\mc{I},
	\end{split}
	\ee		
	and $f(0)=w(0)=f(n+1)=w(n+1)=0$.\footnote{While both $w$ and $f$ are defined over the domain $\{1,\dots,n\}$, we artificially set the non-valid extremum points as $f(0)=w(0)=f(n+1)=w(n+1)=0$, for notational convenience, else, e.g., $f(a+x+1)$ appearing in \eqref{eq:primalvalue} will not be defined for $a+x=n$.}
\end{itemize}
\end{theorem}
Given $f=\fee(w)$, the solution of \eqref{eq:primalvalue} returns both the price of anarchy, and the corresponding worst case instance encoded in $\theta(a,x,b)$ (see the proof in the Appendix). 
Observe that the number of decision variables is $|\mc{I}|\sim\mathcal{O}(n^3)$, while only two scalar constraints are present (neglecting the positivity constraint). The previous program can thus already be solved efficiently. 
 Nevertheless, we are interested in the expression of $\poa(f,w,n)$ (i.e., in the \emph{value} $W^\star$), and therefore consider the dual of \eqref{eq:primalvalue} in Subsection \ref{subsec:duallp}. Before doing so, the next subsection provides intuition on the proof of Theorem~\ref{thm:primalpoa}. 
 \subsection{Outline of Proof}
\label{subsec:informal}
While Equation \eqref{eq:poadef} corresponds to the {\it definition} of price of anarchy, it also describes a (seemingly difficult) \emph{optimization problem}.
The goal of this subsection, is to give an informal introduction on how this optimization problem can be transformed into a finite dimensional LP. %
The non-interested reader can move forward to the next subsection.
We discuss here the case of $f(1)>0$, as showing that $\poa(f,w,n)=0$ whenever $f(1)\le0$ is immediate (see the Appendix).
Additionally, we consider only games $G\in\geefw^n$ with \emph{exactly} $n$ agents. This is without loss of generality, as the price of anarchy over the class of games with $|N|\le n$ agents is the same of that over the class of games with $|N|=n$~agents.\footnote{To see this, note that the price of anarchy of any game with $p$ players $p<n$ can be obtained as the price of anarchy of a corresponding game with $n$ players where we simply set $\mc{A}_i=\emptyset$ for the additional $n-p$ players.}

\vspace*{\vspacesteps}

{\bf \emph{Step 1:}} We observe that the price of anarchy computed over the family of games $G\in\geefw^n$ is the same of the price of anarchy over a reduced family of games, denoted with $\hatgeefw^n$, where the feasible set of every player only contains two allocations: worst-performing equilibrium and optimal allocation, that is $\hat{\mc{A}}_i=\{\ae_i,\aopt_i\}$, %
and definition \eqref{eq:poadef} becomes
\[
\begin{split}
\poa(f,w,n)=&\inf_{G\in \hatgeefw^n}\biggl(\frac{W(\ae)}%
{W(\aopt)}
\biggr)\,,\\
&\quad \text{s.t.}\quad  U_i(\ae)\ge U_i(\aopt_i,\ae_{-i})\quad \forall i\in\N\,,
\end{split}
\]
where we have constrained $\ae$ to be an equilibrium. We do not include the additional constraints requiring $\ae$ to be the \emph{worst-performing} equilibrium and $\aopt$ to provide the \emph{highest} welfare.
Taking the infimum over $\hatgeefw^n$  will ensure this.
\vspace*{\vspacesteps}

\noindent  {\bf \emph{Step 2:}} 
We show that the price of anarchy over the class of games $\hatgeefw^n$ remains unchanged if we introduce the additional constraint $W(\ae)=1$. Thus \eqref{eq:poadef} reduces to 
\be
\begin{split}
	\poa(f,w,n)=&\inf_{G\in \hatgeefw^n}\frac{1}{W(\aopt)}\,,\\
	&\quad \text{s.t.}\quad  U_i(\ae)\ge U_i(\aopt_i,\ae_{-i})\quad \forall i\in\N\,,\\
	&\qquad\quad W(\ae)=1\,.
\end{split}
\label{eq:poastep2informal}
\ee
\revisioned{
\noindent{\bf \emph{Step 3:}} We relax the previous program as in the following 
\be
\label{eq:informalpoa}
\begin{split}
	&\inf_{G\in \hatgeefw^n}\frac{1}{W(\aopt)}\,,\\
	&\quad \text{s.t.}\quad  \sum_{i\in\N}\left(U_i(\ae)- U_i(\aopt_i,\ae_{-i})\right)\ge0\,,\\
	&\qquad~\quad W(\ae)=1\,,
\end{split}
\ee 
where the $n$ equilibrium constraints (one per each player) have been substituted by their sum.
The main difficulty appearing in \eqref{eq:informalpoa} is in how to describe an instance $G\in\hatgeefw^n$ and on how to compute the infimum over all such infinite instances.
To do so, we note that the objective function and the constraints appearing in \eqref{eq:informalpoa} can be encoded using only the parameters $\{\theta(a,x,b)\}_{(a,x,b)\in\mc{I}}$ (see the proof). This trasnforms \eqref{eq:informalpoa} to the program \eqref{eq:primalvalue} appearing in Theorem~\ref{thm:primalpoa}.
\vspace*{\vspacesteps}\\
{\bf \emph{Step 4:}} %
We finally show that the relaxation introduced in Step 3 is tight. Thus, the price of anarchy is the solution of~\eqref{eq:primalvalue}.
}

\subsection{The dual reformulation}
\label{subsec:duallp}Thanks to strong duality, it suffices to solve the dual program of \eqref{eq:primalvalue} to compute $\poa(f,w,n)$. For this purpose,~let
\[
\Ir \coloneqq\{(a,x,b)\in\mc{I}~\text{s.t.}~a\cdotshort x\cdotshort b=0~\text{or}~a+x+b=n\}\,,
\]
and note that $\Ir$ contains all the integer points on the planes $a=0$, $b=0$, $x=0$, $a+x+b=n$ bounding $\mc{I}$.
While the dual program should feature two scalar decision variables and $\mathcal{O}(n^3)$ constraints, the following theorem shows how to reduce the number of constraints to only $|\Ir|=%
{2n^2+1}$. The  goal is to progress towards an \emph{explicit expression} for $\poa(f,w,n)$.

\begin{theorem}[\bf Dual reformulation of \bf $ \rm{\bf PoA}$]
	\label{thm:dualpoa}
	Let $w\in\mb{R}^n_{>0}$ be a welfare basis function, and let $f\in\mb{R}^n$.
	\begin{itemize}
	\item[i)]
	\revisioned{If $f(1)\le0$, then $\poa(f,w,n)=0$ for any $n\in\mb{N}$.} 
	
	\item[ii)] If instead $f(1)>0$, $n\in \mb{N}$, then $\poa(f,w,n)=1/{W^\star}$, where $W^\star$ is the (finite) value of the following program
	\be
	\small
	\begin{split}
		&W^\star = \min_{\lambda\in\mb{R}_{\ge0},\,\mu\in\mb{R}}~ \mu \\[0.1cm]
		&\,\text{s.t.} ~w(b\eqspacezero  +\eqspacezero  x)
		\eqspacezero-\eqspacezero \mu w(a  \eqspacezero+\eqspacezero  x)\eqspacezero+  \eqspacezero\lambda[af(a\eqspacezero + \eqspacezero x)-bf(a  \eqspacezero+ \eqspacezero x\eqspacezero  + \eqspacezero 1)]
		\!\le\! 0\\[0.1cm]
		& \hspace*{60mm}\forall (a,x,b)\in\Ir,
		\label{eq:generalbound}
	\end{split}
	\ee
	and $f(0)=w(0)=f(n+1)=w(n+1)=0$.
	\end{itemize}
\end{theorem}
The proof of the previous theorem (reported in the Appendix) suggests that a further simplification can be made when $f(j)$ is non-increasing for all $j$. In this case the number of constraints reduces to $(n+1)^2-1$, as detailed next.%
\begin{corollary}
	\label{cor:fwnonincreasingpoadual}
	\noindent 
	Let $w\in\mb{R}^n_{>0}$ be a welfare basis function, and let $f\in\mb{R}^n$ with $f(1)>0$.
	\begin{itemize}
	\item[i)]
	If $f(j)$ is non-increasing $\forall j\in[n]$, then $\poa(f,w,n)=1/W^\star$, where
	\be
	\begin{split}
		&W^\star = \min_{\lambda\in\mb{R}_{\ge0},\,\mu\in\mb{R}}~ \mu  \\[0.1cm]
		&\,\text{s.t.}~ \,\mu w(j)   \ge    w(l)+\lambda [j f(j)    -   l f(j   +   1) ]\\
		&~\hspace*{35mm} \forall j,l\in[0, n], \quad 1\le j+l\le n,\\[0.05cm]
		&   \qquad\! \mu w(j)   \ge    w(l)   +   \lambda [(n   -   l) f(j)    -   (n   -   j) f(j   +   1) ]\\
		&~\hspace*{35mm} \forall j,l\in[0, n], \quad~~~~~ j+l> n,\\
	\end{split}
	\label{eq:formulacor1}
	\ee
	and $f(0)=w(0)=f(n+1)=w(n+1)=0$.\vspace*{1mm}
		\item[ii)] If additionally $f(j)\ge \frac{w(j)}{j} \min_{l\in[n]} \frac{l\cdot f(1)}{w(l)}$, then 
	\[
	\lambda^\star=\max_{l\in[n]}\frac{w(l)}{l \cdot f(1)}\,.\]
		\end{itemize}
\end{corollary}
Mimicking the proof of the Corollary \ref{cor:fwnonincreasingpoadual} (see the Appendix), it is possible to obtain a similar result when $f(j)$ is non-decreasing. The result is not included due to space limitations. 
\begin{remark}
	If the optimal value $\lambda^\star$ is known a priori, as in the second statement of Corollary \ref{cor:fwnonincreasingpoadual}, the price of anarchy can be computed \emph{explicitly} from \eqref{eq:formulacor1} as the maximum between $n^2$ real numbers depending on the entries of $f$ and $w$. To see this, divide both sides of the constraints in \eqref{eq:formulacor1} by $w(j)$ for $j\neq 0$. The solution is then found as the maximum of the resulting right hand side, with a corresponding value of %
	\begin{equation}
	\label{eq:mustar} 	
	\medmath{W^\star =\max}
	{
		\begin{cases}
		\medmath{\max_{\substack{1\le j+l\le n \\[0.5mm] j,l\in[0,n], ~j\neq0}}} 
		{\frac{w(l)}{w(j)}+\lambda^\star
		\left[j\frac{f(j)}{w(j)}-l\frac{f(j+1)}{w(j)}\right]}\\[8mm]
		\medmath{\max_{\substack{ j+l> n \\[0.5mm] j,l\in[0,n]}}}
		{\frac{w(l)}{w(j)}+\lambda^\star\left[(n-l)\frac{f(j)}{w(j)}-(n-j)\frac{f(j+1)}{w(j)}\right]}
		\end{cases}}
	\end{equation}
	Equation \eqref{eq:mustar} is reminiscent of the result obtained using a very different approach in \cite[Theorem 6]{marden2014generalized} (limited to Shapley value) and \cite[Theorem 3]{gairing2009covering} (limited to set covering problems). 
	We discuss further connections with these results in Part~II.
\end{remark}

\section{Optimizing the price of anarchy}
\label{sec:optimalpoa}

\label{subsec:optimize}
Given $w$, $n$, and a mechanism $f=\fee(w)$, Theorem \ref{thm:dualpoa} and Corollary \ref{cor:fwnonincreasingpoadual} have reduced the computation of the price of anarchy to the solution of a tractable linear program. Nevertheless, determining the mechanism that maximizes $\poa(f,w,n)$, i.e., devising the best mechanism, is \emph{also a tractable linear program}. The following theorem makes this clear.
\begin{theorem}[\bf Optimizing $ \rm{\bf PoA}$ is a linear program]
	\label{thm:optimizepoa}
	Let $w\in\mb{R}^n_{>0}$ be a welfare basis function, $n\in\mb{N}$. A solution of the design~problem %
	\[
	\argmax_{f\in \mb{R}^n} \poa(f,w,n)
	\]
	is given by the following LP %
	in $n+1$ scalar unknowns%
\be
\label{eq:optf}
	\begin{split}
	&(\fopt,\muopt) \in \argmin_
{
\substack{f\in\mb{R}^n\\ f(1)\ge 1}
,\,\mu\in\mb{R}}
	~ \mu \\[0.1cm]
	&\,\text{s.t.} ~w(b \eqspacetwo   +  \eqspacetwo   x)
	 \eqspacetwo- \eqspacetwo \mu w(a  \eqspacetwo   + \eqspacetwo    x)  \eqspacetwo+af(a \eqspacetwo    +    \eqspacetwo x) \eqspacetwo- \eqspacetwo bf(a    +    \eqspacetwo x    \eqspacetwo +  \eqspacetwo   1) \eqspacetwo
	\le  \eqspacetwo0\\[0.1cm]
& \hspace*{62mm}\forall (a,x,b)\in\Ir,
	\end{split}
\ee
where $f(0)=w(0)=f(n+1)=w(n+1)=0$.
The resulting optimal price of anarchy is $\poa(\fopt,w,n)={1}/{\mu_{\rm opt}}$.
\end{theorem}
The importance of this results stems from its applicability for the \emph{game design procedure} outlined in the introduction.  As a matter of fact, Theorem \ref{thm:optimizepoa} allows to compute the optimal mechanism, for any given welfare basis function, and thus to solve the utility design problem. Applications of these results are presented in Part~II.
 \section{Conclusions}
 Motivated by resource allocation problems arising in multiagent and networked systems, we showed how to provide a priori performance guarantees for distributed algorithms based on a game theoretic approach. With this respect, the paper contains two fundamental results. First, we showed that computing the price of anarchy for the considered class of resource allocation problems (and thus the \emph{approximation ratio} of any algorithm capable of determining a Nash equilibrium) is equivalent to solving a tractable linear program. Second, we showed how to select utility functions so as to maximize such efficiency measure by means of a tractable linear program. In Part II we refine the results derived in this manuscript to the case of submodular, covering, and supermodular problems.

%% file: tex_parts/appendix.tex
\appendices
\section{Proof of Theorem~\ref{prop:limitations}}
 \begin{proof}
$ $\newline
{\bf Claim i).}
In the following we restrict the attention to mechanisms satisfying \eqref{eq:subbudget}, else the smoothness framework would not even apply. We now consider two such mechanisms $\mc{F}$, $\mc{F}'$, and denote with $\{U_i\}_{i\in\N}$, $\{U'_i\}_{i\in\N}$ the utilities \eqref{eq:utilities} obtained with $f_r=\fee (W_r)$, and $f'_r=\fee' (W_r)$, respectively.
We intend to show that if
$
U'_i(a)\le U_i(a)$ for all $i\in N$, $a\in\mc{A}$, $G\in\geefeewee$,
then
\be
\label{eq:lowernotuseful}
\poas(\fee',\wee)\le 
\poas(\fee,\wee).
\ee
To do so, let $G=(N, \ree, \mc{A}, \{W_r\}_{r\in\ree}, \left\{\fee'(W_r)\right\}_{r \in \ree})$ be a game in $\gee_{\fee'}$, and $G=(N, \ree, \mc{A}, \{W_r\}_{r\in\ree}, \left\{\fee(W_r)\right\}_{r \in \ree})$ be the corresponding game in the class of games $\gee_\fee$.
Observe that, if every game $G'\in\gee_{\fee'}$ is smooth with parameters $(\lambda,\mu)$, then every game $G\in\gee_\fee$ is also smooth with the same parameters. Indeed, if $G'$ is $(\lambda,\mu)$-smooth, we have that 
\[
\sum_{i\in\N} U'_i(a'_i,a_{-i})\ge \lambda W(a')-\mu W(a), \quad \forall a,a'\in\mc{A}.
\]
Thus, the corresponding game $G$ is $(\lambda,\mu)$-smooth, since 
\[
\sum_{i\in\N} U_i(a'_i,a_{-i})\ge \sum_{i\in\N} U'_i(a'_i,a_{-i})\ge \lambda W(a')-\mu W(a),
\]
for all  $a,a'\in\mc{A}$, where we used the fact that $U_i(a)\ge U'_i(a)$ for all $i\in\mb{N}$, $a\in\mc{A}$, and $G\in\gee_\fee$. This shows that the set of parameters $(\lambda,\mu)$ feasible for the program defining $\poas(\fee',\wee)$ is also feasible for the program defining $\poas(\fee,\wee)$. It follows that 
\[
\poas(\fee',\wee)\le 
\poas(\fee,\wee).
\]

The above inequality shows that it is never advantageous to lower the utilities from their budget-balanced level, so that the best bound on the price of anarchy achievable using a smoothness argument is attained with a budget-balanced mechanism. \revisioned{In this respect, the unique mechanism satisfying \eqref{eq:subbudget} with equality is $\fee^{\rm es}(W_r)=f_r^{\rm es}$, $f_r^{\rm es}(j)=W_r(j)/j$ for all $j\in\mb{N}$. To see that $\fee^{\rm es}$ satisfies the required property note that 
\[
\sum_{i\in N} U_i(a)=\sum_{i\in N}\sum_{r\in a_i}f_r^{\rm es}(|a|_r)=\!\!\sum_{r\in \cup_i a_i}\!\!|a|_rf_r^{\rm es}(|a|_r)=W(a).
\]
The fact that $\fee^{\rm es}$ is the unique such mechanism is because any budget-balanced mechanism $\fee(W_r)=f_r$ must satisfy 
\[
\sum_{i\in N} U_i(a)=
\sum_{r\in \cup_i a_i}\!\!|a|_r f_r (|a|_r)=
\sum_{r\in \cup_i a_i}\!\!w(|a|_r) = W(a),
\]
for all allocations, for all instances.
Thus, it suffices to consider instances where only one resource is present and shared by all agents. Since the number of agents is arbitrary, this implies
\[
|a|_rf_r(|a|_r)=w(|a|_r)\quad \forall\, |a|_r\in\mb{N}\,,
\]
which is satisfied only by $f_r=\fee^{\rm es}(W_r)$.}

\vspace{\myspaceproofs}
\noindent
{\bf Claim ii).}	
We consider the weighted maximum coverage problem, and first show that 
$\poas(\fee,\wee)\le \poas(\fee^{\rm es},\wee)=1/2.$
Thanks to the previous claim, any mechanism $\fee$ for which the smoothness framework applies (i.e., any mechanism whose corresponding utilities satisfies \eqref{eq:subbudget}) must satisfy $\poas(\fee,\wee)\le\poas(\fee^{\rm es},\wee)$. Thus, we only need to show that $\poas(\fee^{\rm es})=1/2$.
For the mechanism $\fee^{\rm es}$, \cite[Theorem 2]{gairing2009covering} shows that \eqref{eq:smoothcondition} holds with $\lambda=1$ and $\mu =1-1/n$, over all possible instances games where the number of agents is upper bounded by $n$. Since we have not posed any limitations on $n$ until now, taking the limit $n\rightarrow \infty$ gives $\mu=1$, corresponding to $\poas(\fee^{\rm es})\ge 1/2.$
To show that there is no better pair $(\lambda,\mu)$ we show that the price of anarchy is exactly $1/2$. To do so, \revisioned{consider the instance $G$ proposed in \cite[Figure 5]{paccagnan2017arxiv}} and observe that $W(\aopt)=2-1/n$ while $W(\ae)=1$. Taking the limit as $n\rightarrow \infty$ gives $\poas(\fee^{\rm es},\wee)\le \poa(\fee^{\rm es},\wee)\le 1/2$. Since the lower and the upper bound for $\poas(\fee^{\rm es},\wee)$ match, we conclude that $\poas(\fee^{\rm es},\wee)=1/2$.
The second inequality is shown upon observing that $\poa(\fee^\star)=1-1/e$, thanks to \cite{gairing2009covering}.

This concludes the proof, as we have provided a class of problems where $\mc{F}^{\rm es}$ does not optimize the price of anarchy.
\end{proof}
\section*{Proof of Theorem \ref{thm:primalpoa}}
\begin{proof}
The proof makes the steps 1-4 introduced in Subsection \ref{subsec:informal} formal, with particular attention to the steps 3-4. 
As already clarified in the opening of Subsection \ref{subsec:informal}, we consider only games $G\in\geefw^n$ with \emph{exactly} $n$ agents, without loss of generality. 
Additionally, we focus on the case of  $f(1)>0$. The case of $f(1)\le0$ is shown separately in Lemma \ref{lem:f1leq0}.

\vspace*{\myspaceproofs}
\noindent {\bf Step 1.} We intend to show that the price of anarchy computed over $G\in\geefw^n$ is the same of the price of anarchy computed over a reduced set of games.
Consider a game $G\in\geefw^n$ and denote with $\ae$ the corresponding  worst-performing equilibrium (as measured by $W$) and with $\aopt$ an optimal allocation of $G$. 
For every such game $G$, we construct a new game $\hat G$, identical to $G$ in everything but the allocation sets. The allocation sets of the game $\hat{G}$ are defined as $\hat{\mc{A}}_i=\{\ae_i,\aopt_i\}$ for all $i\in\N$, that is, the allocation set of every player in $\hat G$ contains only two allocations: an optimal allocation, and the worst-performing equilibrium of $G$.
Observe that $G$ and $\hat G$ have the same price of anarchy, i.e., 
\[\frac{\min_{a\in \nashe{G}} W(a)}{\max_{a\in\mc{A}} W(a)} = 
\frac{\min_{a\in \nashe{\hat G}} W(a)}{\max_{a\in\hat {\mc{A}}} W(a)}\,.
\] 
With slight abuse of notation we write $\hat G(G)$ to describe the game $\hat G$ constructed from $G$ as just discussed, and with $\hat{\gee}_{f,w}^n$ the class of games $\hat{\gee}_{f,w}^n\coloneqq\{\hat G (G)~\forall G\in\geefw^n\}$.
Observe that $\hat{\gee}_{f,w}^n\subseteq \geefw^n$ (by definition) and since for every game $G\in \geefw^n$, it is possible to construct a game $\hat G \in \hat{\gee}_{f,w}^n$ with the same price of anarchy, it follows that $\poa(f,w,n)$ can be equivalently computed only using games in $\hat{\gee}_{f,w}^n$, i.e.,
\[
\poa(f,w,n)  =\inf_{\hat G\in \hat{\gee}_{f,w}^n}\biggl(\frac{\min_{a\in \nashe{\hat G}} W(a)}{\max_{a\in\mc{A}} W(a)}\biggr)\,.
\]

\noindent{\bf Step 2.} 
\revisioned{
Lemma \ref{lem:lemmapositivewelfare_ateq} ensures that for any game 
$G\in \geefw^n$, and thus for any game
$\hat G\in \hat{\gee}_{f,w}^n$, the equilibrium configuration $\ae$ has strictly positive welfare $W(\ae)$. 
Therefore, for every fixed game $\hat G\in \hat{\gee}_{f,w}^n$, one can construct a corresponding game $\tilde{G}$ identical to $\hat{G}$ in everything but the value of the resources. Each resource $r$ that was associated with a value of $v_r$ in the original game $\hat{G}$ is now associated with a value $v_r/W(\ae)$ in the new game $\tilde{G}$.
Correspondingly, since the welfare has the form
\be
W(a)=\sum_{r\in \cup_{i}a_i} v_r w(|a|_r),
\label{eq:welfareproof}
\ee
 a generic allocation $a\in\mc{A}$ that generated a welfare of $W(a)$ as in \eqref{eq:welfareproof} for the original game $\hat G$, now generates a welfare of $W(a)/W(\ae)$ for the new game $\tilde G$. In particular, the allocation $\ae$ generates a welfare of $1$ for the new game  $\tilde G$. The procedure just introduced simply scales the value of the welfare in \emph{all} the allocations of $\hat{G}$ by the \emph{very same} coefficient. Thus $\ae$ remains the worst-performing Nash equilibrium for the new game $\tilde{G}$. Similarly $\aopt$ remains an optimal allocation for $\tilde{G}$. Therefore the game $\tilde{G}$ must have the same price of anarchy as $\hat{G}$.
In addition, observe that $\{\hat{G}\in \hat{\gee}_{f,w}^n~\text{s.t.}~W(\ae)=1\}\subseteq \hat{\gee}_{f,w}^n$.
Hence, using an identical reasoning as the one carried out in Step 1, it follows that the class of games $\hat G\in \hat{\gee}_{f,w}^n$ has the same price of anarchy of the subclass where we additionally constrain $W(\ae)=1$.} Therefore, the price of anarchy can be computed as
\[
\begin{split}
\poa(f,w,n)=&\inf_{\hat G\in \hat{\gee}_{f,w}^n}
\frac{1}{W(\aopt)}\,,\\
&\quad \text{s.t.}\quad  U_i(\ae)\ge U_i(\aopt_i,\ae_{-i})\quad \forall i\in\N\,,\\
&\qquad\quad W(\ae)=1\,.
\end{split}
\]

\noindent{\bf Step 3.}  
First observe, from the last equation, that $\poa(f)=1/W^\star$, where
\be
\label{eq:originalproofprimal}
\begin{split}
W^\star\coloneqq&\sup_{\hat G\in \hat{\gee}_{f,w}^n}{W(\aopt)}\,,\\
&\quad \text{s.t.}\quad  U_i(\ae)\ge U_i(\aopt_i,\ae_{-i})\quad \forall i\in\N\,,\\
&\qquad\quad W(\ae)=1\,.
\end{split}
\ee
We relax the previous program as in the following
\be
\label{eq:relaxedproofprimal}
\begin{split}
V^\star\coloneqq &\sup_{\hat G\in \hat{\gee}_{f,w}^n}{W(\aopt)}\,,\\
&\quad \text{s.t.}\quad  \sum_{i\in\N} \left(U_i(\ae)- U_i(\aopt_i,\ae_{-i})\right)\ge0\,,\\
&\qquad\quad W(\ae)=1\,,
\end{split}
\ee
where the $n$ equilibrium constraints (one per each player) have been substituted by their sum. 
\revisioned{
We now show that $V^\star$ appearing in  \eqref{eq:relaxedproofprimal} can be computed as 
\be
	\begin{split}
		 V^\star = \max_{\theta(a,x,b)}&\sum_{a,x,b} w(b+x)\theta(a,x,b)\\
		\text{s.t.}&\sum_{a,x,b}[af(a + x)- bf(a + x + 1)]\theta(a,x,b) \ge 0, \\
		&  \sum_{a,x,b}w(a+x)\theta(a,x,b)=1,\\
		& \,\theta(a,x,b)\ge 0,\quad\forall (a,x,b)\in\mc{I}.
	\end{split}
\label{eq:primalproof}
\ee
\noindent Towards this goal, we introduce the variables $\theta(a,x,b)$ with $(a,x,b)\in\mc{I}$.} This parametrization has been used to study covering problems in \cite{ward2012oblivious}, and will be used here to efficiently represent the quantities appearing in \eqref{eq:relaxedproofprimal}.
To begin with, recall that each feasible set is composed of only two allocations, that is $\hat{\mc{A}}_i=\{\ae_i,\aopt_i\}$. For notational convenience, we let 
\begin{itemize}
\item $x_r\in[0, n]$ denote the number of agents selecting resource $r\in\mc{R}$ in both the equilibrium $\ae$ and the optimal allocation $\aopt$, i.e., 
\[x_r=|\{i\in N : r\in \ae_i\cap\aopt_i\}|;\]
\item $a_r+x_r \in [0, n]$ denote the number of agents selecting resource $r\in\mc{R}$ in the equilibrium allocation $\ae$, i.e., 
\[a_r+x_r=|\{i\in N : r\in \ae_i\}|;\]
\item $b_r+x_r \in [0, n]$ denote the number of agents selecting resource $r\in\mc{R}$ in the optimal allocation $\aopt$, i.e.,
\[b_r+x_r=|\{i\in N : r\in \aopt_i\}|.\]
\end{itemize}
Finally, for each tuple $(a,x,b)\in\mc{I}$, we define $\mc{R}(a,x,b)$ as the set containing all the resources $r\in\mc{R}$ that are selected exactly by $a+x$ agents at the equilibrium, $b+x$ at the optimum, of which $x$ agents are selecting $r\in\mc{R}$ both at the equilibrium and at the optimum. Formally, for each $(a,x,b)\in\mc{I}$, let
\[
\mc{R}(a,x,b)\!=\!\{r\!\in\!\mc{R} :~ a_r+x_r\!=\!a + x,~~ b_r+x_r\! =\! b + x, ~~ x_r\!=\!x\}.
\]
Correspondingly, for each $(a,x,b)\in\mc{I}$, we define $\theta(a,x,b)\ge0$ as the sum of the values of the resources in $\mc{R}(a,x,b)$, i.e., 
\[
\theta(a,x,b)= \sum_{r\in \mc{R}(a,x,b)}\!\!\!\!v_r.
\]
In the following we show how these ${(n+1)(n+2)(n+3)/6}$ variables suffice to fully describe the terms appearing in \eqref{eq:relaxedproofprimal}. Indeed, using the notation previously introduced and the definition of the welfare function, it is possible to write
\[
\begin{split}	
W(\ae)&\!=\!\sum_{r\in\mc{R}} \!v_r w(a_r+x_r)=\sum_{a+x=1}^n \!\!\!w(a+x) \!\left(\,\sum_{\substack{r\in \mc{R} \\
a_r+x_r=a+x}
}
\!\!\!\!\!v_r\right)\\
&\!=\!\!\sum_{a,x,b}\!w(a+x)\!\!\left(\sum_{r\in \mc{R}(a,x,b)}\!\!\!\!\!\!v_r\!\right)\!\!=\!\!\sum_{a,x,b} \!w(a+x) \theta(a,x,b),
\\
\end{split}
\]
and similarly for $W(\aopt)$, with $b_r+x_r$ in place of $a_r+x_r$,~i.e.,  
\[
W(\aopt)=\sum_{a,x,b}w(b+x)\theta(a,x,b).
\]
We now move the attention to the terms appearing in the relaxed equilibrium constraint. 
Note that 
$
\sum_{i\in N}U_i(\ae)=\sum_{r\in\mc{R}}v_r (a_r+x_r)f(a_r+x_r).
$
Thus, following the same steps as in the derivation of $W(\ae)$ with $(a_r+x_r)f(a_r+x_r)$ in place of $w(a_r+x_r)$ gives
\[
\sum_{i\in N}U_i(\ae)=\sum_{a,x,b} (a+x) f(a+x) \theta(a,x,b).
\]
Finally, with a similar reasoning we obtain
\[\begin{split}
\sum_{i\in N}\!U_i(\aopt_i,\ae_{-i})&\!=\!\!\sum_{r\in\mc{R}}\!v_r \left[x_rf(a_r+x_r)+b_rf(a_r+x_r+1)\right]
\\
&\!=\!\!\sum_{a,x,b} \!\left[x f(a+x) \!+\! b f(a+x+1)\right]\!\theta(a,x,b).
\end{split}
\]
It follows that $\sum_{i\in\N}\left(U_i(\ae)- U_i(\aopt_i,\ae_{-i})\right)\ge0$ becomes
\[
\begin{split}
&~~~\sum_{i\in\N}\left(U_i(\ae)- U_i(\aopt_i,\ae_{-i})\right)\\
&\!=\!\!\sum_{a,x,b}[(a\eqspacethree+\eqspacethree x)f(a\eqspacethree+\eqspacethree x)\eqspacethree
-\eqspacethree xf(a\eqspacethree+\eqspacethree x)
-\eqspacethree bf(a\eqspacethree+\eqspacethree x\eqspacethree+\eqspacethree 1)]\theta(a,x,b) \\
&\!=\!\!\sum_{a,x,b}\left[af(a+x)-bf(a+x+1)\right]\theta(a,x,b)\ge0\,.
\end{split}
\]
Substituting these expressions in \eqref{eq:relaxedproofprimal}, one gets 
\[
	\begin{split}
		 V^\star = \sup_{\theta(a,x,b)}&\sum_{a,x,b} w(b+x)\theta(a,x,b)\\
		\text{s.t.}&\sum_{a,x,b}[af(a + x)- bf(a + x + 1)]\theta(a,x,b) \ge 0, \\
		&  \sum_{a,x,b}w(a+x)\theta(a,x,b)=1,\\
		& \,\theta(a,x,b)\ge 0,\quad\forall (a,x,b)\in\mc{I}.
	\end{split}
	\]
To transform the latter expression in \eqref{eq:primalproof} it suffices to show that the $\sup$ is attained. To see this observe that the objective function is continuous and that the decision variables $\{\theta(a,x,b)\}_{(a,x,b)\in\mc{I}}$ live in a compact space. Indeed $\theta(a,x,b)$ is constrained to the positive orthant for all $(a,x,b)\in\mc{I}$.
Additionally, each decision variable with $a+x\neq0$ is upper bounded due to the constraint $W(\ae)=1$,~i.e.,
\[
\sum_{\substack{(a,x,b)\in\mc{I}\\a+x\ge1}}w(a+x)\theta(a,x,b)=1\,,
\]
where $w(j)\neq0$ by assumption. Finally, the decision variables left, i.e. those of the form $\theta(0,0,b)$, $b\in[n]$ are upper bounded due to the equilibrium constraint, which can be rewritten as
\[
\sum_{b\in[n]}\!
bf(1)\theta(0,0,b)\le
\!\!\!
\sum_{\substack{(a,x,b)\in\mc{I}\\a+x\ge1}}
\!\!\!\!
[af(a+x)-bf(a+x+1)]\theta(a,x,b),
\]
where $f(1)\neq0$ by assumption. \revisioned{This proves that $V^\star$ defined in \eqref{eq:relaxedproofprimal} is equal to the value of the program in \eqref{eq:primalproof}, and shows, in addition, that the value $W^\star$ is finite.}

\vspace*{\myspaceproofs}
\noindent 
\revisioned{
{\bf Step 4.} Observe that $V^\star \ge W^\star$ by definition of $V^\star$. In addition, Lemma \ref{lemma:relaxedmatchesoriginal} shows that $V^\star \le W^\star$, so that $V^\star = W^\star$. Thus, the price of anarchy is $\poa(f,w,n)=1/W^\star$ where $W^\star$ is equal to the value of the program in \eqref{eq:primalproof}, which is identical to the the desired expression for $W^\star$ in \eqref{eq:primalvalue}.}

\end{proof}

\begin{lemma}
\label{lem:f1leq0}
 For any welfare basis $w$, if $f(1)\leq0$,  it holds	%
	\begin{equation}
	\poa(f,w,n) = 0, \ \forall n \in\mb{N}.
	\end{equation}
\end{lemma}
	
\begin{proof}
	Consider a game with one agent, with resource set $\ree = \{r_1, r_2\}$ where $v_{r_1} > 0$ and $v_{r_2} = 0$. Now suppose the action set of agent $1$ is $\aee_1 = \{\{r_1\}, \{r_2\}\}$.  Regardless of the specific value of $f(1)$, $a_1 = \{r_2\}$ is an equilibrium.  Hence, the price of anarchy of this specific game is $W_{r_2}(1)/W_{r_1}(1) = 0$. Consequently the price of anarchy over the class $\geefw^n$, that is $\poa(f,w,n)$, must be zero for any $n\ge 1$.
\end{proof}

\begin{lemma}
\label{lem:lemmapositivewelfare_ateq}
Assume $f(1)>0$. For any game $G\in\geefw^n$, it is
\[
W(\ae)>0\text{~~for all~~} \ae\in\rm{NE}(G)\,.
\]
\end{lemma}
\begin{proof}
Let us consider a fixed game $G\in \geefw^n$.
By contradiction, let us assume that $W(\ae)=0$ for some $\ae\in\rm{NE}(G)$.
It follows that all the players must have distributed themselves on resources that are either valued zero, or have selected the empty set allocation (since $w(j)>0$ when $j\ge 1$). Hence, their utility function must also evaluate to zero.
However, Definition \ref{ass} ensures that $W(\aopt)>0$. Thus, there must exists a player $p$ and a resource $r$ with $v_r>0$ contained in one of the allocations belonging to $\mc{A}_p$ (else we would have $W(\aopt)=0$). Observe that no other player is currently selecting this resource, otherwise it would be $W(\ae)>0$. If player $p$ was to deviate and select the allocation containing $v_r$, his utility would be strictly positive (since $f(1)>0$). Thus $\ae$ is not an equilibrium: a contradiction.
Repeating the same reasoning for all games $G\in \geefw^n$ yields the claim.
\end{proof}

\begin{lemma}
\label{lemma:relaxedmatchesoriginal}
Consider $W^\star$ as in \eqref{eq:originalproofprimal} and $V^\star$ as in \eqref{eq:primalproof}. It holds that $V^\star \le W^\star$.
\end{lemma}
\begin{proof}
	\revisioned{For any $\{\theta(a,x,b)\}_{(a,x,b)\in \mc{I}}$ feasible solution of \eqref{eq:primalproof}, we will construct an instance of game $\hat G$ satisfying the constraints of \eqref{eq:originalproofprimal} too. This allows to conclude that $V^\star\le W^\star$.} 
		
	\vspace*{\myspaceproofs}
	Consider $\{\theta(a,x,b)\}_{(a,x,b)\in \mc{I}}$ a feasible tuple for \revisioned{\eqref{eq:primalproof}} with value $v$. For every $(a,x,b)\in \mc{I}$ and for each $j\in [n]$ we create a resource $r(a,x,b,j)$ and assign to it the value of $\theta(a,x,b)/n$, i.e., $v_{r(a,x,b,j)}=\theta(a,x,b)/n$, $\forall j\in[n]$. We then construct the game $\hat G$ by defining $\hat{\mc{A}}_i=\{\ae_i,\aopt_i\}$ for all $i\in N$, where the resources are assigned as follows 
	\[
	\begin{split}
	\ae_i\!=\!&\cup_{j=1}^n \!\{r(a,x,b,j)~\text{s.t.}~a+x\!\ge\! 1+((j\!-\!i)\,\mathrm{mod}\,n)\}\,,\\
	\aopt_i\!=\!&\cup_{j=1}^n \!\{r(a,x,b,j)~\text{s.t.}~b+x\!\ge\! 1 +((j\!-\!i\!+\!b)\,\mathrm{mod}\,n)\}.
	\end{split}
	\]
	Informally this corresponds to the following construction: for a fixed tuple $(a,x,b)$, position the $n$ resources $r(a,x,b,j)$ indexed by $j\in [n]$ on a circle, as in Figure \ref{fig:wheel1and2}.
	As part of the equilibrium allocation $\ae_i$, agent $i$ adds $a+x$ of these resources, starting from the resource $r(a,x,b,j)$ with $j=i$ and moving clockwise.
	As part of the optimum allocation $\aopt_i$, agent $i$ adds a total of $b+x$ resources, starting from the resource $1+(i-1-b)\,\mathrm{mod}\,n$ and moving clockwise. Repeat the above construction running over all possible $(a,x,b)\in\mc{I}$, where for each new $(a,x,b)\in\mc{I}$ new resources are added to the agent's equilibrium and optimal allocations, according to the process just described.

\begin{figure}[h!]
\centering
\includegraphics[scale=1]{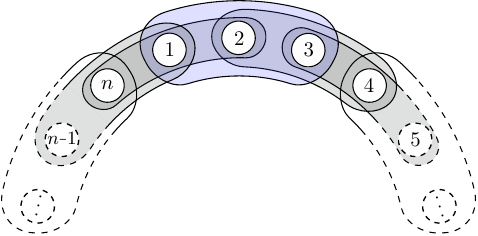}
\includegraphics[scale=1]{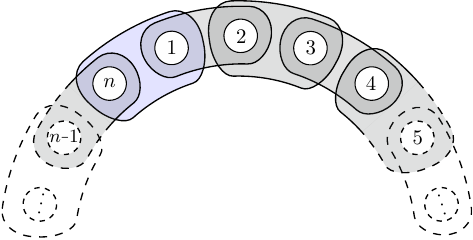}
\caption{
Construction of the agent's equilibrium and optimal allocations. Representation of the resources corresponding to $a=2$, $x=1$, $b=1$. In the figures above the $n$ resources $\{r(2,1,1,j)\}_{j\in[n]}$ are denoted by white circles, marked with the corresponding index $j$, and positioned on a ring. 
As part of the equilibrium allocation (top figure) each agent adds $a+x=3$ resources. In particular, agent $1$ adds the resources with $j\in\{1,2,3\}$ (represented by the blue set); agent $2$ adds the resources with $j\in\{2,3,4\}$; agent $3$ adds the resources with $j\in\{3,4,5\}$, etc.
As part of the optimal allocation (bottom figure) each agent adds $b+x=2$ resources. In particular, agent $1$ adds the resources with $j\in\{n,1\}$ (represented by the blue set); agent $2$ adds the resources with $j\in\{1,2\}$; agent $3$ adds the resources with $j\in\{2,3\}$, etc.
Observe how \emph{all} resources are selected by \emph{exactly} $a+x=3$ (resp. $b+x=2$) agents at the equilibrium (resp. optimum). Finally note that $x=1$ resources are shared between the equilibrium and the optimal allocation of each agent. 
}
\label{fig:wheel1and2}
\end{figure}

	We begin by showing $W(\ae)=1$ and $W(\aopt)=v$.
	Observe that for any fixed resource $r_{(a,x,b,j)}$ (i.e. for every fixed tuple $(a,x,b,j)$), there are exactly $a+x$ (resp. $b+x$) players selecting $r_{(a,x,b,j)}$ in the equilibrium allocation (resp. optimum). It follows that
	\[\begin{split}
	W(\ae)&=\sum_{j\in[n]}\sum_{a+x>0} v_{r(a,x,b,j)}w(a+x)\\
	      &=\sum_{j\in[n]}\sum_{a+x>0} \frac{\theta(a,x,b)}{n}w(a+x)\\
	      &=\sum_{a,x,b}w(a+x)\theta(a,x,b)=1\,,
	\end{split}
	\]
	With an identical reasoning, one shows that 
	\[
	\begin{split}
	W(\aopt)&=\sum_{j\in[n]}\sum_{b+x>0} v_{r(a,x,b,j)}w(b+x)\\
	      &=\sum_{a,x,b}w(b+x)\theta(a,x,b)=v\,.
	\end{split}
	\]
Finally, we prove that $\ae$ is indeed an equilibrium, i.e. it satisfies 
$U_i(\ae)-U_i(\aopt_i,\ae_{-i})\ge 0$ for all $i\in\N$. Towards this goal, we recall that the game under consideration is a congestion game with potential function $\phi:\mc{A}\rightarrow\mb{R}_{\ge0}$
\[
\phi(a)=\sum_{r\in
\cup_i a_i}\sum_{\k=1}^{|a|_r} v_r f(\k),
\]
see \cite{rosenthal1973class}.	It follows that $U_i(\ae)-U_i(\aopt_i,\ae_{-i})=\phi(\ae)-\phi(\aopt_i,\ae_{-i})$ and so we equivalently prove that 
	\[
	\phi(\ae)-\phi(\aopt_i,\ae_{-i})\ge 0\quad\forall i\in\N\,.
	\]
Thanks to the previous observation, according to which every resource $r_{(a,x,b,j)}$ is covered by exactly $a+x$ players at the equilibrium, we have
\[
\begin{split}
\phi(\ae)=&\sum_{r\in \cup_i{\ae_i}} \!\!\!v_r \sum_{\k=1}^{a+x}f(\k)=\sum_{j\in [n]}\sum_{a,x,b}\frac{\theta(a,x,b)}{n}\sum_{\k=1}^{a+x}f(\k)\\
		 =&\frac{1}{n}\sum_{a,x,b}n\,{\theta(a,x,b)}\sum_{\k=1}^{a+x}f(\k)\,.
\end{split}
\]
When moving from $\ae$ to $(\aopt_i,\ae_{-i})$ there are $b$ resources selected by one extra agent and $a$ resources selected by one less agent. The remaining $n-a-b$ resources are chosen by the same number of agents as in the equilibrium $\ae$. 
Thus,
\[\small\begin{split}
&\phi(\ae)-\phi(\aopt_i,\ae_{-i})=\\
=&\frac{1}{n}\sum_{a,x,b}n\,{\theta(a,x,b)}\sum_{ \k=1}^{a+x}f( \k)+\\
&-\frac{1}{n}\!\sum_{a,x,b}\!{\theta(a,x,b)}\!
\eqspacefour
\left[
b\!\!\!\sum_{ \k=1}^{a+x+1}\!\!\!\!f( \k)\eqspacefour +\eqspacefour a\!\!\!\!\sum_{ \k=1}^{a+x-1}\!\!\!\!f( \k)\eqspacefour+\eqspacefour
(n\eqspacefour -\eqspacefour a\eqspacefour -\eqspacefour b)\!\sum_{ \k=1}^{a+x}\!f( \k)
\!\right]\\[0.2cm]
=&\frac{1}{n}\sum_{a,x,b}\!{\theta(a,x,b)}\left[
\!(a+b)\!\sum_{ \k=1}^{a+x}f( \k) -a\!\!\sum_{ \k=1}^{a+x-1}\!\!\!f( \k)
-b\!\sum_{ \k=1}^{a+x+1}\!\!\!f( \k)\right]\\[0.2cm]
=&\frac{1}{n}\sum_{a,x,b}\theta(a,x,b)[af(a+x)-bf(a+x+1)]\ge0\,,
\end{split}
\] 
where the inequality holds because $\theta(a,x,b)$ is assumed feasible for \eqref{eq:primalproof}. This concludes the proof. 
\end{proof}

\section*{Proof of Theorem \ref{thm:dualpoa}}
\begin{proof}
The case of $f(1)\le0$ was already shown in the proof of Theorem \ref{thm:primalpoa}. Thus, we restrict to $f(1)>0$, and divide the proof in two parts: \emph{Part a)} writing the dual of the original program in \eqref{eq:primalvalue}; \emph{Part b)} showing that only the constraints obtained for $(a,x,b)\in\Ir$ are binding.
	
	\vspace*{\myspaceproofs}
	\noindent{\bf Part a).} Upon stacking the decision variables $\theta(a,x,b)$ in the vector $y\in\mb{R}^{\ell}$, ${\ell=(n+3)(n+2)(n+1)/6}$, and after properly defining the coefficients $c$, $d$, $e\in\mb{R}^{\ell}$, the program \eqref{eq:primalvalue} can be compactly written as 
\[
\begin{split}
W^\star = &\max_{y} \, c^\top y\\
&~\text{s.t.}\quad -e^\top y\le 0\,, \quad\, (\lambda)\\
&\qquad\, d^\top y -1 = 0\,, \quad (\mu)\\
&\qquad~~\,\quad -y\le 0\,. \quad \,(\nu)
\end{split}
\]
The Lagrangian function is defined for $\lambda\ge0$, $\nu\ge0$ as  $\mathcal{L}(y,\lambda,\mu,\nu)=c^\top y-\lambda(-e^\top y)-\mu(d^\top y -1)-\nu^\top(-y)=(c^\top+\lambda e^\top+\nu-\mu d^\top)y+\mu$, while the dual function reads as
\[
g(\lambda,\mu,\nu) = \mu \quad \text{if}\quad c^\top+\lambda e^\top+\nu^\top-\mu d^\top=0\,,
\]
and it is unbounded elsewhere.
Hence the dual program takes the form 
\[\begin{split}
&\min_{\lambda\in\mb{R}_\ge0,\,\mu\in\mb{R}}~ \mu \\
&~~~~~\text{s.t.}\quad c+\lambda e-\mu d\le0\,,
\end{split}
\]
which corresponds, in the original variables, to%
\be
\label{eq:proofdual}
\begin{split}
	&W^\star = \min_{\lambda\in\mb{R}_{\ge0},\,\mu\in\mb{R}}~ \mu \\[0.1cm]
	&\,\text{s.t.} ~w(b\eqspace  +\eqspace  x)
	\eqspace-\eqspace \mu w(a  \eqspace+\eqspace  x)\eqspace+  \eqspace\lambda[af(a\eqspace  + \eqspace x)-bf(a  \eqspace+ \eqspace x\eqspace  + \eqspace 1)]
	\eqspace\le\eqspace 0\\[0.1cm]
	& \hspace*{60mm}\forall (a,x,b)\in\mc{I}\,.
\end{split}
\ee
By strong duality\footnote{The primal LP \eqref{eq:primalvalue} is always feasible, since $\theta(0,1,0)=1/w(1)$, $\theta(a,x,b)=0$ $\forall\,(a,x,b)\in\mc{I}\setminus(0,1,0)$ satisfies all the constraints in \eqref{eq:primalvalue}}, the value of \eqref{eq:primalvalue} matches \eqref{eq:proofdual}. Finally, observe that the dual is attained since the primal value is finite.

\vspace*{\myspaceproofs}
\noindent{\bf Part b).} In this step we show that only the constraints with $(a,x,b)\in\Ir$ are necessary in \eqref{eq:proofdual}, thus obtaining \eqref{eq:generalbound}. 

Observe that when $(a,x,b)\in\mc{I}$ and $a+x=0$, $b$ can take any value $1\le b\le n$, and these indices are already included in $\Ir$. 
Similarly for the indices $(a,x,b)\in\mc{I}$ with $b+x=0$.
Thus, we focus on the remaining constraints, i.e. those with $a+x\neq 0$ and $b+x\neq 0$.
We change the coordinates from the original indices $(a,x,b)$ to $(j,x,l)$, $j\coloneqq a+x$, $l\coloneqq b+x$. 
The constraints in \eqref{eq:proofdual} now read as 
\be
\begin{split}
\mu w(j)&\ge w(l) + \lambda[(j - x)f(j) - (l - x)f(j + 1)]\\
& =  w(l)+\lambda[jf(j)-lf(j+1) +{ x(f(j+1)-f(j))}]\,,
\end{split}
\label{eq:reducetobound}
\ee	
where $(j,x,l)\in \hat{\mc{I}}
$ and 
$
\hat{\mc{I}}=\{(j,x,l)\in \mathbb{N}_{\ge0}^3~\text{s.t.}~ 1\le j-x+l\le n,~ j\ge x,~ l\ge x,~ j,l\neq 0\}.
$
In the remaining of this proof we consider $j$ fixed, while $l,~x$ are free to move within $\hat{\mc{I}}$. This corresponds to moving the indices in the rectangular region defined by the blue and green patches in Figures \ref{fig:pyramid_decreasing}, \ref{fig:pyramid_increasing}.\\
Observe that for $j=n$ it must be $l=x$ (since $-x+l\le 0$ and $l-x\ge 0$), i.e., in the original coordinates $b=0$, which represents the segment on the plane $b=0$ with $a+x=n$. These indices already belong to $\Ir$. %
Thus, we consider the case $j\ne n$ and divide the reasoning in two parts. 
\setlist[itemize]{leftmargin=*}
\begin{itemize}
\item {Case of $f(j+1)\le f(j)$.}\\[0.1cm]
In the following we fix $l$ as well (recall that we have previously fixed $j$). This corresponds to considering points on a black dashed line on the plane $j=\rm{const}$ in Figure \ref{fig:pyramid_decreasing}.
 The term $f(j+1)-f(j)$ is non-positive and so the most binding constraint in \eqref{eq:reducetobound} is obtained picking $x$ as small as possible.  Since it must be $x\ge 0$ and $x\ge j+l-n$, for fixed $j$ and $l$, we set $x = \max\{0,j+l-n\}$. In the following we show that these constraints are already included in $\Ir$.
\begin{itemize}
\item[-] When $j+l\le n$, i.e., when $a+b+2x\le n$, it is $x=\max\{0,j+l-n\}=0$. These indices correspond to points on the plane $x=0$, ($1\le a+b\le n$)  and so they are already included in $\Ir$ (white diamonds in Figure \ref{fig:pyramid_decreasing}).
\item[-] When $j+l> n$, i.e., when $a+b+2x> n$, it is $x=\max\{0,j+l-n\}=j+l-n$, i.e., $a+b+x=n$. These indices correspond to points on the plane $a+b+x=n$, which are included in $\Ir$ too (black circles in Figure~\ref{fig:pyramid_decreasing}).
\end{itemize}
\begin{figure}[h!]
\centering
\includegraphics[width=0.4\textwidth]{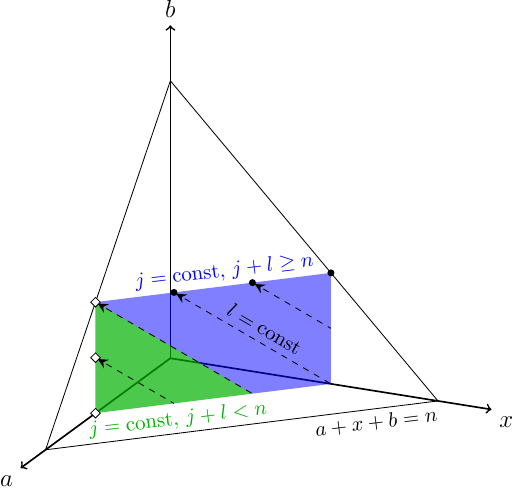}
\caption{Indices representation for case a).}
\label{fig:pyramid_decreasing}
\end{figure}

\item {Case of $f(j+1)>f(j)$}. \\[0.1cm]
In the following we fix $l$ as well (recall that we have previously fixed $j$). This corresponds to considering points on a dashed black line on the plane plane $j=\rm{const}$ in Figure \ref{fig:pyramid_increasing}).
The term $f(j+1)-f(j)$ is positive and so the most binding constraint in \eqref{eq:reducetobound} is obtained picking $x$ as large as possible. Since it must be $x\le l$, $x\le j$ and $x\le j+l-1$, we set $x = \min\{j,l\}$. In the following we show that these constraints are already included in \eqref{eq:generalbound}.
 
\begin{itemize}
\item[-] When $j\le l$ i.e. when $a\le b$, it is $x = \min\{j,l\}=j$, i.e., $a=0$. These indices correspond to points on the plane $a=0$, ($1\le x+b\le n$) which are included in $\Ir$ (white diamonds in Figure \ref{fig:pyramid_increasing}).
\item[-] When $j>l$, i.e., when $a>b$, it is $x = \min\{j,l\}=l$, i.e., $b=0$. These indices correspond to points on the plane $b=0$, ($1\le a+b\le n$) which are included in $\Ir$~too (black circles in Figure \ref{fig:pyramid_increasing}).
\end{itemize}
\begin{figure}[ht!]
\centering
\includegraphics[width=0.4\textwidth]{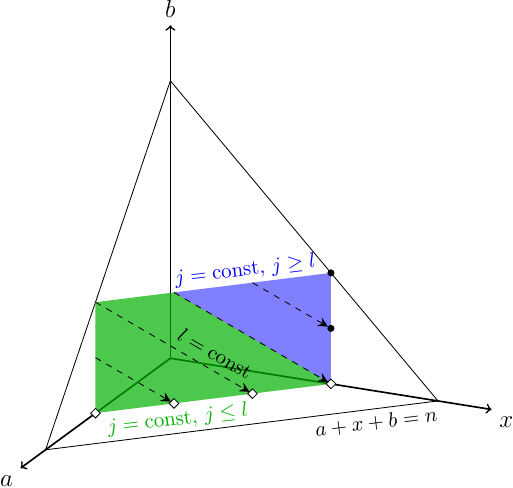}
\caption{Indices representation for case b).}
\label{fig:pyramid_increasing}
\end{figure}
\end{itemize}
\end{proof}
\section*{Proof of Corollary \ref{cor:fwnonincreasingpoadual}}
\begin{proof}
$ $\\
	{\bf Claim i).} Following the proof of Theorem \ref{thm:dualpoa} (Part b),   we note that if $f(j)$ is non-increasing for all $j\in[n]$, the only binding indices are those lying on the the two surfaces $x=0$, $1\le a+b\le n$ and $a+x+b\le n$. 
The surface $x=0$, $1\le a+b\le n$ gives 
\be
\mu w(j)\ge w(l)+\lambda [j f(j)-l f(j+1)]
\label{eq:part1reduced}
\ee
for $1\le j+l\le n$ and $j,l\in[0,n]$,  where we have used the same change of coordinates of the proof of Theorem \ref{thm:dualpoa} i.e. $j=a+x$, $l=b+x$.
The surface $a+x+b=n$ gives
\[
\mu w(n-b)=w(n-a)+\lambda[af(n-b)-bf(n-b+1)]\,,
\]
which can be written as 
\be
\mu w(j)\ge w(l)+\lambda [(n-l) f(j)-(n-j) f(j+1)]\,,
\label{eq:part2reduced}
\ee
for $j+l> n$ and $j,l\in[0,n]$, where we have used $j=a+x=n-b$, $l=b+x=n-a$.
 Thus, we conclude that \eqref{eq:part1reduced} and \eqref{eq:part2reduced} are sufficient to describe the constraints in \eqref{eq:generalbound}.

\vspace*{\myspaceproofs}
\noindent {\bf Claim ii).}
We first note that as a consequence of the assumption 
\[
f(j)\ge \frac{w(j)}{j} \min_{l\in[n]} \frac{l\cdot f(1)}{w(l)},
\]
it follows that $f(j)>0$ for all $j\in[n]$, since $w(l)>0$, $f(1)>0$. Thus we need not worry about the case of $f(j)=0$ for $j\in[n]$ in the remainder of the proof.

When $j=0$ the constraints yield $\lambda\ge w(l)/(lf(1))$ for $l\in[n]$. Define
	\[\lambda^\star =\max_{l\in[n]} \frac{w(l)}{l f(1)}\,,
	\] and observe that any feasible $\lambda$ must satisfy $\lambda\ge \lambda^\star$.  These constraints correspond to straight lines parallel to the $\mu$ axis. 
	
To prove the claim, we show that the most binding constraints amongst all those in \eqref{eq:formulacor1}  with $j\neq 0$ are of the form $\mu\ge \alpha \lambda +\beta$, where $\alpha\ge 0$, i.e., the most binding constraints are straight lines in the $(\lambda,\mu)$ plane pointing north-east. Thus, the value of $\lambda$ that minimizes $\mu$ is attained with $\lambda$ as small as possible, i.e., $\lambda=\lambda^\star$. See Figure \ref{fig:lambdamu} for an illustrative plot. 	
	
\begin{figure}[h!]
\centering
\hspace*{-4mm}
\includegraphics[scale=.8]
{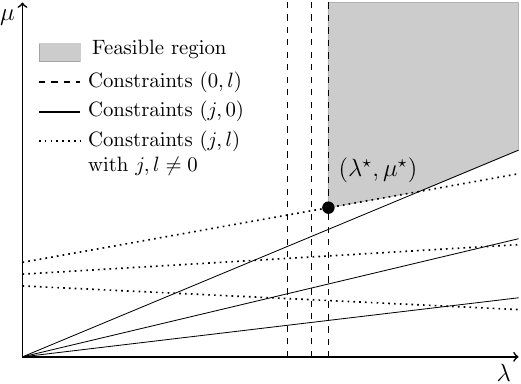}
\vspace*{-3mm}
\caption{The three classes of constraints used in the proof of Corollary \ref{cor:fwnonincreasingpoadual}.}
\label{fig:lambdamu}
\end{figure}

We now consider the case of $l=0$, for which the constraints yield $\mu \ge \lambda j f(j)/w(j)$, with $j\in[n]$.
	These constraints are straight lines pointing north-east in the $(\lambda,\mu)$ plane due to $f(j)/w(j)\ge 0$ for $j\ge 1$. We are thus left to check the constraints with $j\neq0$ and $l\neq 0$.
	
	To do so, we prove that if one such constraint (identified by the indices $(j,l)$) has negative slope, the constraint identified with $(j,0)$ is more binding. 
	This will conclude the proof, since the constraint $(j,0)$ has non-negative slope, as just seen.
	We split the reasoning depending on wether $1\le j+l\le n$ or $j+l> n$ as the constraints in \eqref{eq:formulacor1} have a different expression.
	
	$\bullet$ Case of $1\le j+l\le n$: the constraints read as 
	\[
	\mu \ge \frac{w(l)}{w(j)}+\frac{\lambda}{w(j)}[jf(j)-lf(j+1)].
	\]
	Observe that the case $j=n$, corresponds to $l=0$ since $j+l\le n$, and this case has already been discussed. Thus, we consider the case $j\neq n$ in the following.
	To complete the reasoning we assume that the above constraint has negative slope, that is
    $jf(j)-lf(j+1)<0$, and show that the constraint $(j,0)$ is more binding, i.e., that
	\[
	\lambda j\frac{f(j)}{w(j)}\ge \frac{w(l)}{w(j)}+\lambda j\frac{f(j)}{w(j)}-\lambda l\frac{f(j+1)}{w(j)}\,,
	\]	
	which is equivalent to showing 
	\be
	w(l)-\lambda lf(j+1)\le 0\,.
	\label{eq:toshow1}
	\ee
	Since $jf(j)-lf(j+1)<0$, it follows that
	\[
	l>j\frac{f(j)}{f(j+1)}\ge j\,,
	\]
	by non-increasingness and positivity of $f(j)$. Therefore $l\ge j+1$.
	Consequently, $w(l)\le (w(l)f(j+1))/f(l)$, again by non-increasingness and positivity of $f(l)$.
Using this, we can bound the left hand side of \eqref{eq:toshow1} as \\
\[
w(l)-\lambda l f(j+1) \le \left(\frac{w(l)}{f(l)}-\lambda l\right)f(j+1)\,.
\]
By assumption $f(j)\ge \frac{w(j)}{j} \min_{l\in[n]} \frac{l\cdot f(1)}{w(l)}$ for all $j\in[n]$. Setting $j=l$ gives 
$w(l)/f(l)\le  l\cdot\max_{l\in[n]} w(l)/(lf(1))= l \lambda^\star$. Therefore we conclude that \eqref{eq:toshow1} holds, since 
\[
w(l)-\lambda l f(j+1) \!\le\!\! \left(\!\frac{w(l)}{f(l)}-\lambda l\!\right)\!f(j+1)\!\le\! (\lambda^\star - \lambda)l f(j+1)\!\le\!0,
\]
where the last inequality holds since $f(j+1)>0$ and $\lambda\ge \lambda^\star$ for every feasible~$\lambda$.

	$\bullet$ Case of $j+l> n$:  the constraints read as
	\[
	\mu w(j)\ge w(l)+\lambda [(n-l) f(j)-(n-j) f(j+1)]\,.
	\]
	Observe that if $j=n$, then the above constraint has non-negative slope since $l\le n$ and $f(j+1)=0$. Thus, in the following we consider the case of $j\neq n$. 
	To complete the proof we assume that the above constraint has negative slope, that is $(n-l) f(j)-(n-j) f(j+1)<0$, and show that the constraints $(j,0)$ is more binding, i.e., that 
	\[
	\lambda j \frac{f(j)}{w(j)} \ge 
	\frac{w(l)}{w(j)}+\lambda (n-l) \frac{f(j)}{w(j)}-\lambda (n-j) \frac{f(j+1)}{w(j)}
	\]
	or equivalently 
	\be
	w(l)+\lambda(n-l-j)f(j)-\lambda(n-j)f(j+1)\le 0.
	\label{eq:toshow2}
	\ee
	Observe that
	\[\begin{split}
	&\lambda(n-l-j)f(j)-\lambda(n-j)f(j+1)=\\
	&=-l\lambda f(j)+\lambda (n-j)(f(j)-f(j+1))\\
	&\le -l \lambda f(j)+l\lambda (f(j)-f(j+1))\\
	&=-l\lambda f(j+1)\,,
	\end{split}
	\]
	where the inequality holds because $f(j)-f(j+1)\ge0$ and $n-j< l$ since we are considering indices with $j+l>n$. Thus the left hand side of \eqref{eq:toshow2} is upper bounded by 
	\[
	w(l)+\lambda(n-l-j)f(j)-\lambda(n-j)f(j+1)
	\le w(l)-l\lambda f(j+1).
	\]
	Therefore in the following we equivalently show that $w(l)-l\lambda f(j+1)\le 0$.
	Since $(n-l) f(j)-(n-j) f(j+1)<0$, it must be 
	\[
	n-j > \frac{f(j)}{f(j+1)}(n-l)> n-l,
	\]
	by non-increasingness and positivity of $f(j)$. Thus it must be \mbox{$l\ge j+1$}. Since the equation we are left to show, i.e.,  $w(l)-l\lambda f(j+1)\le 0$, is identical to \eqref{eq:toshow1}, and since again $l\ge j+1$, we can apply the same reasoning as in the case of $1\le j+l \le n$ and conclude.
\end{proof}
\section*{Proof of Theorem \ref{thm:optimizepoa}}
\begin{proof}
For given $f\in\mb{R}^n$, if $f(1)\le 0$, then $\poa(f,w,n)=0$ (see Theorem \ref{thm:primalpoa}), while if $f(1)> 0$ then $\poa(f,w,n)>0$ (consequence of the fact that $W^\star$ in \eqref{eq:primalvalue} if finite, see Theorem \ref{thm:primalpoa}).
Thus, any $f$ with $f(1)\le0$ can not be optimal. Therefore, in the following we consider only $f\in\mb{R}^n$ with $f(1)>0$.
In addition, Lemma \ref{lem:rescalingdoesnotchange} shows that the price of anarchy does not change upon scaling $f$ with a positive constant. Thus, without loss of generality,  we consider only mechanisms $f=\fee(w)$ satisfying $f\in F$, where
\[
F = \{f:[n]\rightarrow \mb{R}~\text{s.t.}~f(1)\ge 1\}.
\]
For any such $f\in F$, the price of anarchy can be computed using \eqref{eq:generalbound}. Therefore, devising a mechanism that maximizes the price of anarchy is equivalent to determining $f\in F$ minimizing $W^\star$ defined in \eqref{eq:generalbound}, i.e., 
\be
\label{eq:optimalproof}
	\begin{split}
		&\argmin_{f\in F} \min_{\lambda\in\mb{R}_{\ge0},\,\mu\in\mb{R}}~ \mu \\[0.1cm]
		&\,\text{s.t.} \,w(b\eqspace  +\eqspace  x)
		\eqspace-\eqspace \mu w(a  \eqspace+\eqspace  x)\eqspace+  \eqspace\lambda[af(a\eqspace  + \eqspace x)-bf(a  \eqspace+ \eqspace x\eqspace  + \eqspace 1)]
		\eqspace\le\eqspace 0\\[0.1cm]
		& \hspace*{60mm}\forall (a,x,b)\in\Ir\,.
	\end{split}
\ee
Lemma \ref{lem:finite} shows that the latter program is well posed, in the sense that minimum is attained for some $f\in F$ with bounded components. 

The program in \eqref{eq:optimalproof} is non linear, but the decision variables $\lambda$ and $f$ always appear multiplied together. Thus, we define $\tilde f(j)\coloneqq\lambda f(j)$ for all $j\in [0, n+1]$ and observe that the constraint obtained in \eqref{eq:generalbound} for $(a,x,b)=(0,0,1)$ gives $
\tilde f(1)=\lambda f(1)\ge 1$, which also implies $\lambda \ge 1/f(1)>0$ since $f(1)>0$.
Folding the $\min$ operators gives
\be
	\begin{split}
	&(\tilde{f}_{\rm opt},\mu_{\rm opt}) \in \argmin_{ \tilde f(1)\ge 1,\,\mu\in\mb{R}}~ \mu \\[0.1cm]
		&\,\text{s.t.} ~w(b\eqspace  +\eqspace  x)
		\eqspace-\eqspace \mu w(a  \eqspace+\eqspace  x)\eqspace+  \eqspace a\tilde{f}(a\eqspace  + \eqspace x)-b\tilde{f}(a  \eqspace+ \eqspace x\eqspace  + \eqspace 1)
		\eqspace\le\eqspace 0\\[0.1cm]
		& \hspace*{60mm}\forall (a,x,b)\in\Ir\,.
\label{eq:generalbounddualproof}
	\end{split}
\ee
Finally, observe that $\tilde f_{\rm opt}$ is also feasible for the original program, since $\tilde f_{\rm opt}(1)\ge 1$.
Additionally, we note that $\tilde f_{\rm opt}$ and $f_{\rm opt}$ give the same price of anarchy (since $\tilde f_{\rm opt}=\lambda_{\rm opt} \fopt$, with $\lambda_{\rm opt}>0$, see Lemma \ref{lem:rescalingdoesnotchange}). Thus $\tilde{f}_{\rm opt}$ solving \eqref{eq:generalbounddualproof} must be optimal and $\poa(\fopt)=1/\mu_{\rm opt}$.
\end{proof}
\begin{lemma}
\label{lem:rescalingdoesnotchange}
For any welfare basis $w$, any mechanism $f=\fee(w)$, and any $n\in\mb{N}$, the price of anarchy $\poa(f,w,n)$ is invariant by scaling $f$ with any positive constant $\alpha >0$, i.e.,
\[
\poa(f,w,n)=\poa(\alpha \cdot f,w,n).
\]
\end{lemma}
\begin{proof}
For any given game $G\in\geefw^n$, consider the corresponding game $\hat G \in\gee_{\alpha \cdot f,w}^n$ that is identical to $G$ in everything, but employs $\alpha \cdot f$ in place of $f$. Observe that if $\aopt$ is an optimal allocation for $G$, this must also be an optimal allocation for $\hat G$, since modifying $f$ does not impact the definition of the welfare, or the allocation sets.
Additionally, observe that any allocation $\ae$ that is an equilibrium for the game $G$ is also an equilibrium for the game $\hat {G}$, and vice-versa. This is because, whenever $\ae$ satisfies the equilibrium condition for $G$, i.e., 
\[
\sum_{r\in \ae_i} v_r f(|\ae|_r)\!\ge\!\!  \sum_{r\in a_i} v_r f(|(a_i, \ae_{-i})|_r),~~ \forall a_i\in\mc{A}_i,~ i\!\in\! N
\]
it also satisfies the equilibrium conditions for $\hat{G}$, i.e., 
\[
\sum_{r\in \ae_i}\!\!v_r \alpha \cdot f(|\ae|_r)\!\ge\!\!  \sum_{r\in a_i} \!\!v_r  \alpha \cdot f(|(a_i, \ae_{-i})|_r),~~ \forall a_i\in\mc{A}_i,~i\!\in\! N
\]
and vice-versa, due to the fact that multiplying by $\alpha >0$ does not change the sign of the inequalities. Thus the games $G$ and $\hat G$ have the same price of anarchy. Repeating the reasoning over all games $G\in\geefw^n$ gives the claim for the whole class.
\end{proof}

\begin{lemma}
\label{lem:finite}
The minimum appearing in \eqref{eq:optimalproof} is attained by some mechanism $\fee(w)=f\in F$.
\end{lemma}
\begin{proof}
In the following, we show that the infimum
\be
	\begin{split}
		&\inf_{f\in F} \min_{\lambda\in\mb{R}_{\ge0},\,\mu\in\mb{R}}~ \mu \\[0.1cm]
		&\,\text{s.t.}\,w(b\eqspace  +\eqspace  x)
		\eqspace-\eqspace \mu w(a  \eqspace+\eqspace  x)\eqspace+  \eqspace\lambda[af(a\eqspace  + \eqspace x)-bf(a  \eqspace+ \eqspace x\eqspace  + \eqspace 1)]
		\eqspace\le\eqspace 0\\[0.1cm]
		& \hspace*{60mm}\forall (a,x,b)\in\Ir\,
	\end{split}
\label{eq:infimumattained}
\ee
is attained by some mechanism with \emph{bounded components}, and therefore by $f\in F$. 
Since the price of anarchy of anarchy of $f$ and of $\alpha \cdot f$ with $\alpha>0$ is the same (see Lemma \ref{lem:rescalingdoesnotchange}), \eqref{eq:infimumattained} is equivalent to 
\be
	\begin{split}
		&\inf_{\substack{f\in F\\f(1)=1}} \min_{\lambda\in\mb{R}_{\ge0},\,\mu\in\mb{R}}~ \mu \\[0.1cm]
		&\,\text{s.t.}\,w(b\eqspace  +\eqspace  x)
		\eqspace-\eqspace \mu w(a  \eqspace+\eqspace  x)\eqspace+  \eqspace\lambda[af(a\eqspace  + \eqspace x)-bf(a  \eqspace+ \eqspace x\eqspace  + \eqspace 1)]
		\eqspace\le\eqspace 0\\[0.1cm]
		& \hspace*{60mm}\forall (a,x,b)\in\Ir\,.
	\end{split}
	\label{eq:infimum2}
\ee
Therefore, in the following we consider a given $f\in F$ with $f(1)=1$, and construct from it $f_M$ as follows: $f_M (j)=M$, with $M\in\mb{R}$ for some fixed $j \in [2, n]$, while $f_M$ exactly matches $f$ for the remaining components. We then show that there exists $M^{+}\ge0$ such that $\poa(f_M,w,n)<\poa(f,w,n)$ for any $M\ge M^+$. Similarly, we show that there exists $M^{-}\le0$ such that $\poa(f_M,w,n)<\poa(f,w,n)$ for any $M\le M^-$. 
 Thus $f_M$ can not attain the infimum for $M\ge M^+$ or $M\le M^-$ as the corresponding $f$ would give a better price of anarchy. 
Repeating this reasoning for any $f\in F$ with $f(1)=1$, and for any possible $j\in[2,n]$ one concludes that the distribution rule achieving the infimum in \eqref{eq:infimum2} can not be unbounded along a single direction. With an identical argument, one can show that the distribution rule achieving the infimum in \eqref{eq:infimum2} can not be unbounded along two or more directions simultaneously. This provides the desired result.

\vspace*{\myspaceproofs}
To conclude we show that $\exists M^+\ge0$, $\exists M^-\le0$ such that $\poa(f_M,w,n)<\poa(f,w,n)$ for all $M\ge M^+$ and for all $M\le M^-$. To do so, observe that the price of anarchy of $f\in F$ is $\poa(f,w,n)=1/W^\star$, where $W^\star$ is the solution to the primal problem in \eqref{eq:primalvalue}. As shown in Theorem \ref{thm:primalpoa}, it is $W^\star <+\infty$ and so $\poa(f,w,n)>0$ strictly. On the other hand, thanks to Theorem \ref{thm:dualpoa}, the price of anarchy of $f_M$ can be computed for any $M$ as $\poa(f_M,w,n)=1/W_M^\star$, where
	\[\small
	\begin{split}
		&W_M^\star = \min_{\lambda\in\mb{R}_{\ge0},\,\mu\in\mb{R}}~ \mu \\[0.1cm]
		&\,\text{s.t.} ~w(b\eqspace  +\eqspace  x)
		\eqspace-\eqspace \mu w(a  \eqspace+\eqspace  x)\eqspace+  \eqspace\lambda[af_M(a\eqspace  + \eqspace x)-bf_M(a  \eqspace+ \eqspace x\eqspace  + \eqspace 1)]
		\eqspace\le\eqspace 0\\[0.1cm]
		& \hspace*{60mm}\forall (a,x,b)\in\Ir.
	\end{split}
	\]
First, observe that any feasible $\lambda$ satisfies ${\lambda\ge \frac{1}{f_M(1)}=1}$, else the constraints obtained form the previous linear program with $a=x=0$, $b=1$ would be infeasible. Further, consider the constraints with $b=0$, $x=0$, $a=j\ge 2$. They amount to
\[
\mu\ge \lambda j f_M(j)\ge \frac{j}{f_M(1)} f_M(j) = {jM}\,,
\]
so that
\[
\poa(f_M)=\frac{1}{W_M^\star} \le \frac{1}{jM}\,.
\]
Thus, it is possible to make $\poa(f_M,w,n)$ arbitrarily close to zero, by selecting $M$ sufficiently large, i.e., $\exists M^+\ge 0$ such that 
$\poa(f_M,w,n)<\poa(f,w,n)$ for all $M\ge M^+$, since $\poa(f,w,n)$ is bounded away from zero, as argued above. 
Similarly, consider the constraints $a=0$, $b=1$, $x=j-1 \ge 1$
\[
\mu \ge \frac{w(j)-\lambda f_M(j)}{w(j-1)},
\]
from which be obtain 
\[
\mu \ge 
\frac{w(j)}{w(j-1)} - \frac{\lambda M}{w(j-1)} \stackrel{(M< 0)}{\ge} 
\frac{w(j)}{w(j-1)} - \frac{M}{w(j-1)}
,
\]
where the last inequality holds only for a negative value of $M$, since $\lambda \ge 1$.
Thus, it is possible to make the term $- \frac{M}{w(j-1)}$ arbitrarily large by selecting $M$ to be a large negative number, so that $\poa(f_M,w,n)$ is arbitrarily close to zero. Formally, $\exists M^-\le 0$ such that 
$\poa(f_M,w,n)<\poa(f,w,n)$ for all $M\le M^-$, since $\poa(f,w,n)$ is bounded away from zero.
A similar reasoning applies for a given mechanism with two or more unbounded components.
\end{proof}

%% file: tex_parts/biography.tex
\begin{IEEEbiography}
[{\includegraphics[width=1in,height=1.25in,clip,keepaspectratio]{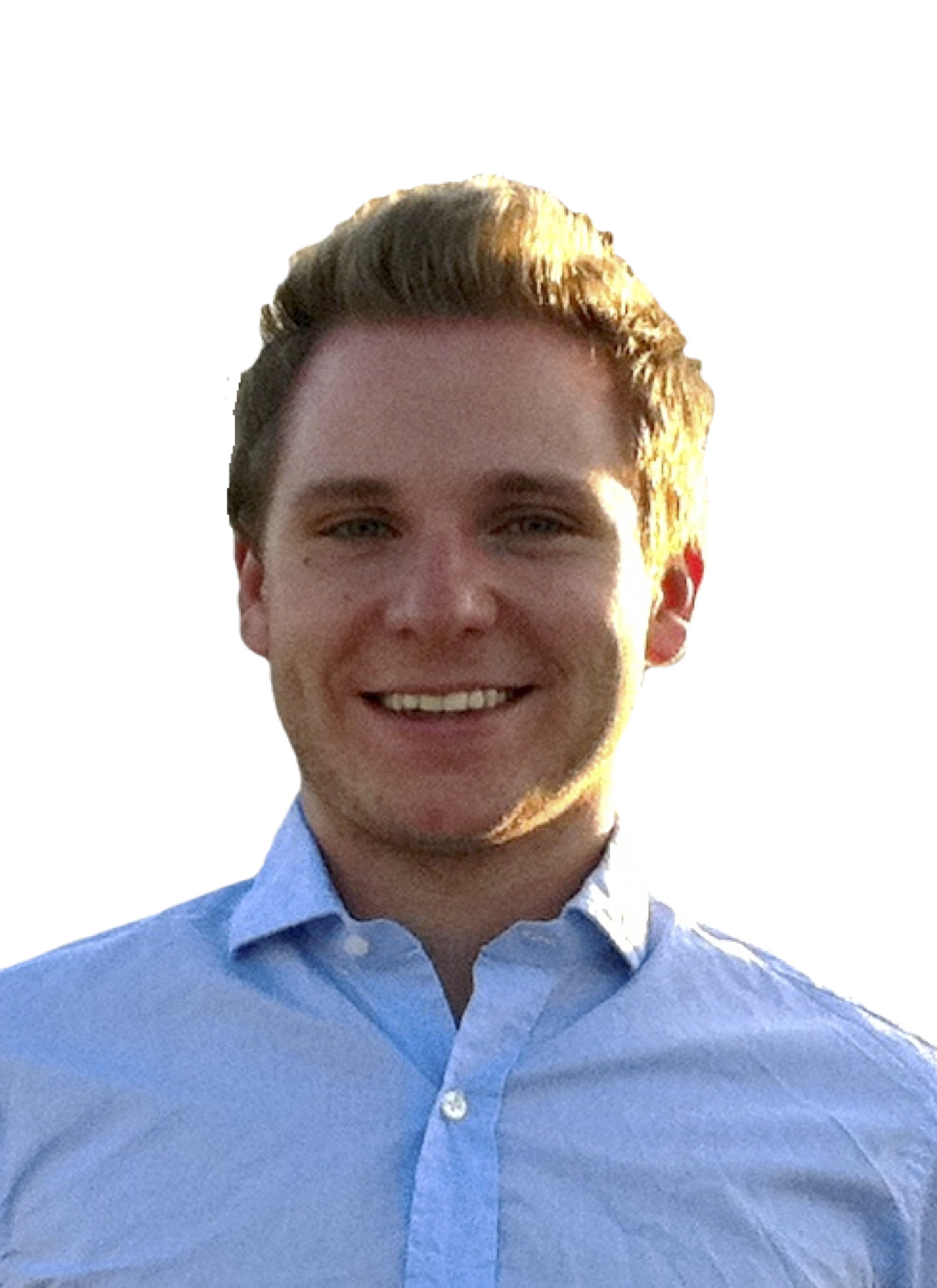}}]
{Dario Paccagnan} is a Postdoctoral Fellow with the Mechanical Engineering Department and the Center for Control, Dynamical Systems and Computation, University of California, Santa Barbara.
In 2018 Dario obtained a Ph.D. degree from the Information Technology and Electrical Engineering Department, ETH Z\"{u}rich, Switzerland. He received his B.Sc. and M.Sc. in Aerospace Engineering in 2011 and 2014 from the University of Padova, Italy. In 2014 he also received the M.Sc. in Mathematical Modelling from the Technical University of Denmark; all with Honours. 
Dario was a visiting scholar at the University of California, Santa Barbara in 2017, and at Imperial College of London, in 2014.
He is recipient of the SNSF fellowship for his work in Distributed Optimization and Game Design. His research interests are at the interface between distributed control and game theory with applications to multiagent systems and smart cities.
\end{IEEEbiography}

\begin{IEEEbiography}
[{\includegraphics[width=1in,height=1.25in,clip,keepaspectratio]{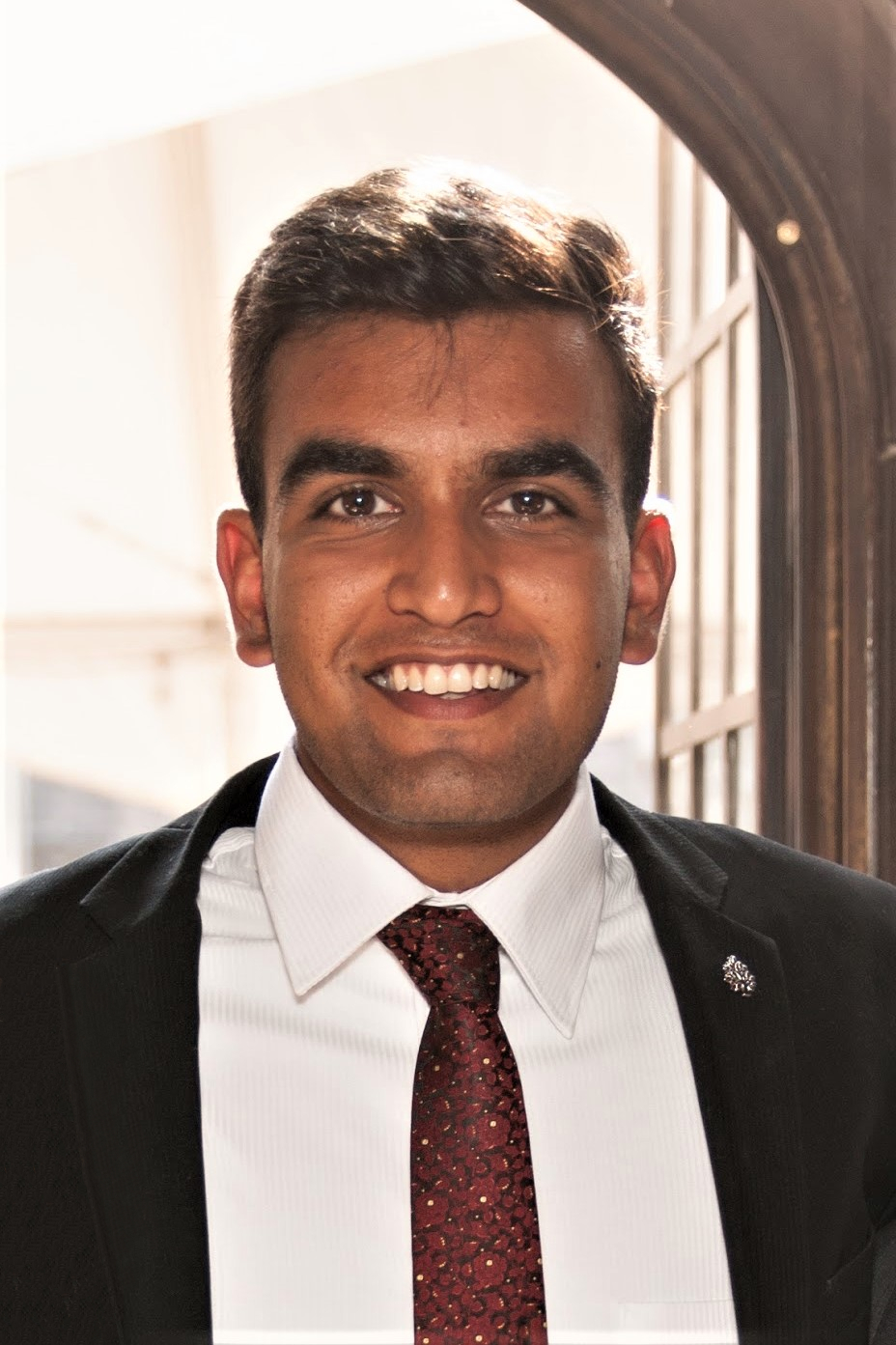}}]
{Rahul Chandan} is a PhD student in the Electrical and Computer Engineering Department at the University of California, Santa Barbara since September 2017. He received his BASc in Electrical and Computer Engineering from the University of Toronto in June 2017. Rahul's research interests lie in the application of game theoretic and classical control methods to the analysis and control of multiagent systems.
\end{IEEEbiography}

\begin{IEEEbiography}
[{\includegraphics[width=1in,height=1.25in,clip,keepaspectratio]{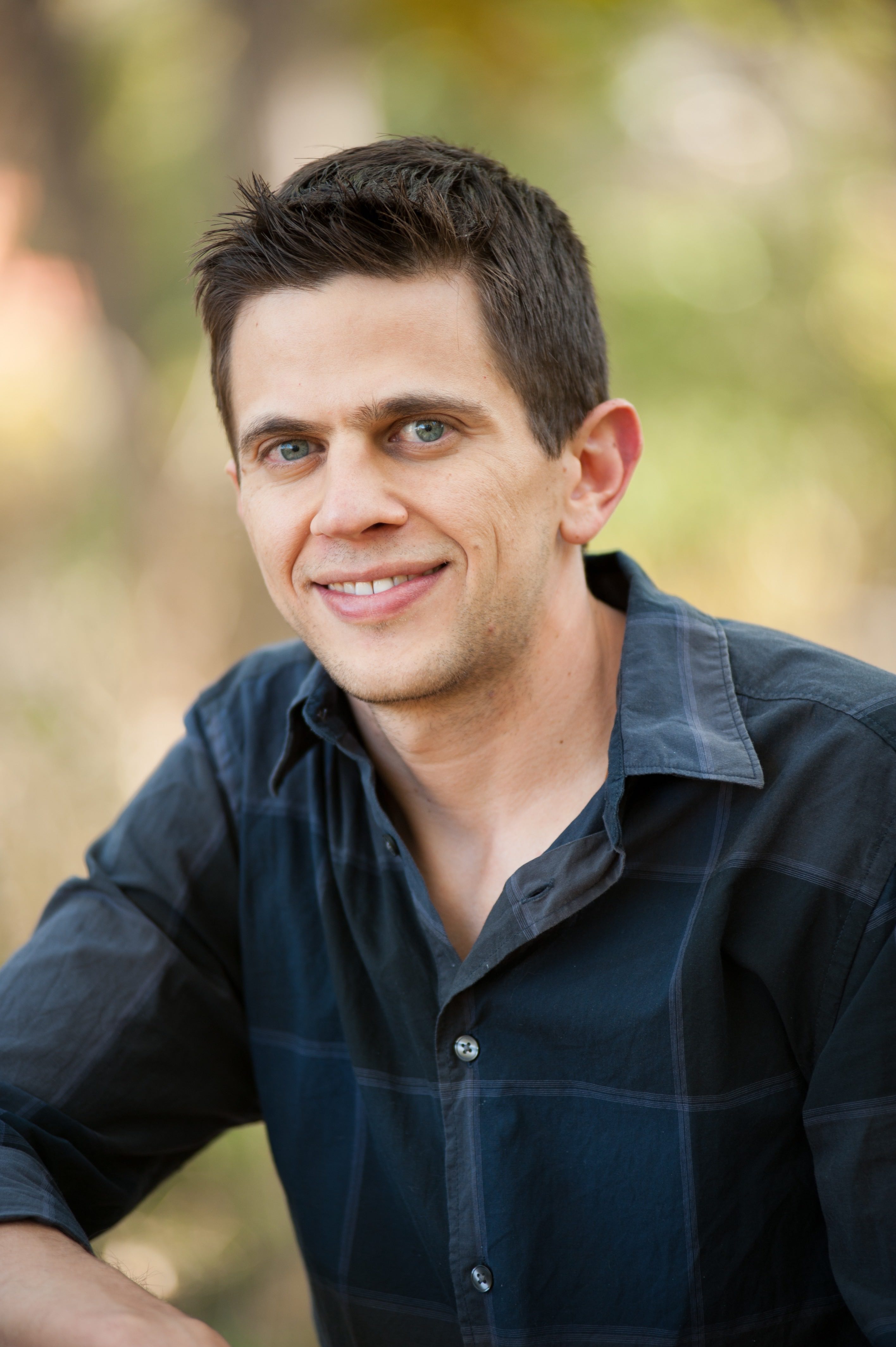}}]
{Jason Marden} is an Associate Professor in the Department of Electrical and Computer Engineering at the University of California, Santa Barbara. Jason received a BS in Mechanical Engineering in 2001 from UCLA, and a PhD in Mechanical Engineering in 2007, also from UCLA, under the supervision of Jeff S. Shamma, where he was awarded the Outstanding Graduating PhD Student in Mechanical Engineering. After graduating from UCLA, he served as a junior fellow in the Social and Information Sciences Laboratory at the California Institute of Technology until 2010 when he joined the University of Colorado. Jason is a recipient of the NSF Career Award (2014), the ONR Young Investigator Award (2015), the AFOSR Young Investigator Award (2012), the American Automatic Control Council Donald P. Eckman Award (2012), and the SIAG/CST Best SICON Paper Prize (2015). Jason's research interests focus on game theoretic methods for the control of distributed multiagent systems.
\end{IEEEbiography}